\documentclass{amsart}
\usepackage{amsmath,amsthm,amsfonts}
\usepackage{amssymb}
\usepackage{hyperref}
\usepackage{times}
\usepackage{paralist}
\usepackage{graphicx}
\usepackage[round]{natbib}

\newtheoremstyle{myplain}
{}
{}
{\itshape}
{}        
{\scshape}
{}
{.5em}
{}

\newtheoremstyle{myremark}
{}
{}
{}
{}        
{\scshape}
{}
{.5em}
{}

\usepackage{color}

\usepackage[T1]{fontenc}

\theoremstyle{myplain}
\newtheorem{thm}{Theorem}\numberwithin{thm}{section}
\newtheorem{cor}[thm]{Corollary}
\newtheorem{lem}[thm]{Lemma}
\newtheorem{prop}[thm]{Proposition}
\theoremstyle{definition}

\newtheorem{ass}[thm]{Assumption}

\theoremstyle{remark}
\newtheorem{rem}[thm]{Remark}

\numberwithin{equation}{section}

\newcommand{\Real}{\mathbb R}

\newcommand{\F}{\mathcal{F}}

\newcommand{\prob}{\mathbb{P}}
\newcommand{\qprob}{\mathbb{Q}}

\newcommand{\expec}{\mathbb{E}}

\newcommand{\limT}{\lim_{T\rightarrow \infty}}

\newcommand{\indic}{\mathbb{I}}
\newcommand{\pare}[1]{\left(#1\right)}
\newcommand{\bra}[1]{\left[#1\right]}

\newcommand{\wt}[1]{\widetilde{#1}}

\begin{document}

\title{Robust Portfolios and Weak Incentives \endgraf in Long-Run Investments}
\thanks{
\endgraf
The authors are grateful to Jak{\v{s}}a Cvitani{\'c}, Ren Liu, and an anonymous referee for their useful comments. \endgraf
Paolo Guasoni is partially supported by the ERC (278295), NSF (DMS-1109047), SFI (07/MI/008,
07/SK/M1189, 08/SRC/FMC1389), and FP7 (RG-248896). Hao Xing is supported in part by STICERD at London School of Economics. Johannes Muhle-Karbe gratefully acknowledges partial support by the National Centre of Competence in Research ``Financial Valuation and Risk Management'' (NCCR FINRISK), Project D1 (Mathematical Methods in Financial Risk Management), of the Swiss National Science Foundation (SNF)}
\author{Paolo Guasoni}
\address{Boston University, Department of Mathematics and Statistics, 111 Cummington Street Boston, MA 02215, USA, and Dublin City University, School of Mathematical Sciences, Glasnevin, Dublin 9, Ireland.}
\email{guasoni@bu.edu}
\author{Johannes Muhle-Karbe}
\address{ETH Z\"urich, Departement f\"ur Mathematik, R\"amistrasse 101, CH-8092, Z\"urich, Switzerland, and Swiss Finance Institute}
\email{johannes.muhle-karbe@math.ethz.ch}
\author{Hao Xing}
\address{London School of Economics and Political Science, Department of Statistics, 10 Houghton St, London WC2A 2AE, UK}
\email{h.xing@lse.ac.uk}

\begin{abstract}
When the planning horizon is long, and the safe asset grows indefinitely, isoelastic portfolios are nearly optimal for investors who are close to isoelastic for high wealth, and not too risk averse for low wealth. We prove this result in a general arbitrage-free, frictionless, semimartingale model.
As a consequence, optimal portfolios are robust to the perturbations in preferences induced by common option compensation schemes, and such incentives are weaker when their horizon is longer. Robust option incentives are possible, but require several, arbitrarily large exercise prices, and are not always convex.
\end{abstract}

\maketitle

\section{Introduction}
Investors pursue long-term goals both by managing their portfolios and by designing incentives, such as stock and option grants, to align managers' actions with their interests.
This paper explores the implications of long-term investment for both portfolio choice and incentive contracts, overcoming their traditional separation in the literature, and some puzzling results that this separation has generated.

The main message in the literature for long-term portfolio choice comes from turnpike theorems:\footnote{See \cite{mossin1968optimal},
\cite{leland1972turnpike},
\cite{ross1974portfolio},
\cite{hakansson1974convergence},
\cite{MR736053},
\cite{cox1992continuous},
\cite{MR1629559},
\cite{MR1805320},
\citet*{dybvig1999portfolio},
\citet*{detemple2010dynamic}, and
\citet*{guasoni.al.11}.}
when the horizon is distant, optimal portfolios depend only on preferences at high levels of wealth, hence a generic investor should invest like an isoelastic investor with the same asymptotic risk aversion.\footnote{Here and henceforth, ``isoelastic'' refers to preferences with constant relative risk aversion.} These theorems are welcome news for long-term investors who seek simple portfolio allocation strategies, as they imply that local differences in preferences are irrelevant at long horizons, and allow to focus on the one-parameter family of isoelastic strategies.

However, turnpike theorems also have unsettling implications for managers' compensation. When investors offer incentive stock options to a manager, the incentive is a perturbation of the manager's preferences. If turnpike theorems hold, they imply that incentive schemes lose their strength as their horizon (i.e., vesting period) increases. The central goal of this paper is to understand the extent and the limits of this effect -- and its potential remedies.

Our main theorem provides conditions on preferences under which, for a long horizon, an isoelastic portfolio is approximately optimal for a generic investor. This is not another turnpike theorem: while turnpikes show that generic portfolios are close to isoelastic portfolios \emph{when the horizon is distant}, we show that the welfare loss to an investor who uses an isoelastic portfolio \emph{throughout the entire period} is negligible -- in \emph{relative monetary terms}. That is, the fraction of wealth lost by using a simple isoelastic portfolio rather than the exact, but more complex, optimal portfolio, declines to zero as the planning horizon increases. The theorem holds in a general semimartingale market without arbitrage opportunities, in which the safe asset grows indefinitely.

To appreciate the strength of this statement, compare this welfare effect with that of a small reduction in the interest rate. Compounded over a long horizon, a lower interest rate translates into a large relative decline in wealth, as the ratio between investments growing at different rates diverges. From this viewpoint, isoelastic portfolios are surprisingly robust -- and incentives remarkably weak.

The message of this result is as positive for portfolio choice as it is negative for typical incentive schemes (those based on options with strike price near the money): if, at long horizons, the optimal strategy for an investor is robust to local perturbations in preferences, then, by the same token, a manager's policy is insensitive to common stock and option grants, which modify preferences only locally. Thus, to be effective at long horizons, an incentive contract must modify preferences at levels much higher than current wealth.

We argue that incentive contracts based on options with several, arbitrarily high strike prices are robust to the horizon. A simple example is a (multiple of a) contract that pays one option for each strike price in a regularly spaced grid: its payoff is approximately the underlying assets's squared price, and the value of the contract is finite in most common models.\footnote{In Black and Scholes' model, the price of the contract $e^{-\sigma^2 T} S_T^2$ is $S_0^2$, i.e., the square of the current stock price. This contract is replicated by a portfolio with equal weights in call options of all strikes. }

Our main result is a natural complement to the extant turnpike literature. The latter compares, as the horizon increases, optimal portfolios for generic investors to optimal portfolios of isoelastic investors with the same horizon, and establishes joint restrictions on markets and preferences, under which these portfolios are close. Put differently, the question is whether the optimizers of the two maximization problems are similar at the early stages of the investment process. A natural practical question (but, somewhat unnaturally, not addressed in the literature) is whether these maximizers are good substitutes: if a generic investor chooses to use an isoelastic portfolio over the entire investment horizon, then under which conditions will utility be close to optimal, and in which sense? Although such a result is plausible, there are two potential pitfalls.

First, turnpike theorems only assert that optimizers are close at the beginning, not necessarily at the end, of the period. Thus, switching a portfolio with another may have dreadful results if they grow apart as time passes -- a concrete possibility, as even optimal isoelastic portfolios shift between the beginning and the end of the period.\footnote{\cite{Guasoni-Robertson} show examples in which the long-run optimal portfolio has utility equal to negative infinity, if used over the entire period.} Second, the generic investor may have nearly isoelastic preferences at high wealth levels, but may be far more risk averse for very low wealth, and it is unclear whether such bad states have a negligible effect, even in the long run. This concern is especially relevant when the isoelastic portfolio is very risky.

The main result of the present study clarifies these issues in a general setting. It also helps to reconcile some of the puzzling, or at least counterintuitive results in the literature on executive compensation. \citet{jensen1976theory} already recognized that \emph{managers of large publicly held corporations seem to behave in a risk averse way to the detriment of the equity holders}, as  projects with positive net present value may be foregone to avoid additional risk. Stock-based compensation attempts to align managers' interests with those of shareholders, but may also backfire, leading managers to engage in hedging activities, as predicted by \citet{amihud1981risk} as well as \citet{smith1985determinants}, and confirmed empirically by \citet{may1995managerial}. As a remedy, option-based incentives aim at rewarding managerial risk-taking by introducing convexity in their payoffs. The intuition, which finds its roots in the familiar notion that prices of options with convex payoffs increase with volatility, is that \emph{the asymmetric payoffs of call options make it more attractive for managers to undertake risky projects.} \citep{defusco1990effect}. However, \cite{Carpenter} and \cite{ross2004compensation} show that this intuition is misleading: with a nontradeable option, the manager will focus on its private reservation value (that is, the certainty equivalent) rather than on its risk-neutral value, and the effect of convex incentives is ambiguous in general. In a full-information model with risk-sharing, \citet*{cadenillas2007optimal} characterize optimal contracts, and also find that they may be either convex or concave. The numerical results of \citet*{larsen2005optimal} also lead to a similar conclusion.

Adding to the debate, \cite{hall2000optimal} argue that the standard practice of awarding stock options with exercise price equal to the stock price at grant date, potentially motivated by a favorable accounting treatment, is suboptimal, and find that optimal exercise prices are generally higher.
\citet{cadenillas2004leverage} also find that optimal exercise prices increase with manager performance and firm size. \citet{edmans2012dynamic} note the dynamic problems with option compensation: \emph{securities given to incentivize the CEO may lose their power over time: if firm value declines, options may fall out of the money and bear little sensitivity to the stock price}. In fact, this issue motivates the common industry practice of ``repricing'', that is resetting the exercise price after a sharp drop in the stock price \citep{acharya2000optimality, chen2004executive}, a custom that has attracted criticism for rewarding poor performance and weakening original incentives.

Our results clarify the interplay between the exercise price and the horizon in option compensation. If the exercise price is held constant, the option loses its effect as the horizon increases, in that its certainty equivalent becomes arbitrarily small compared to the manager's wealth. By contrast, an incentive contract with multiple exercise prices remains effective even after large changes in the stock price, and for long horizons.

Multiple exercise prices also shed light on the ambiguous effects of option grants, and convex incentives in general, on managerial risk-taking. The intuition is that \citep{carr2001optimal} a portfolio of call and put options can recreate any regular function of the underlying, including incentive contracts of power type $x^\alpha$, $\alpha>0$.
If the manager's risk aversion is high, the incentives contracts that reduce it correspond to $0<\alpha<1$, hence are \emph{concave}, not convex. By contrast, a manager with low risk aversion is motivated to take risks by a package of call options with all strikes, and this convex incentive is robust to changes in stock prices and to long horizons.

Finally, note that this paper focuses on the effects of compensation contracts based on combinations of cash, stock, and options, which are prevalent in practice, but it does not investigate their optimality for a possible principal.
(See \citet*{bolton2005contract} and \citet*{Cvitanic-Zhang} for recent surveys in contract theory.) On the other hand, our results allow considerable flexibility in both investment opportunities and preferences, shedding new light on the effect of the horizon in typical compensation contracts.

The rest of the paper is organized as follows: Section 2 introduces the model, and presents the main result and its implications. Section 3 lays the groundwork for the proof of the main result, recalling the general duality results of \cite{Bouchard-Touzi-Zeghal} and some auxiliary results that will be used repeatedly in the sequel. The main result is proved for power utilities ($p \neq 0$) in Section 4, and for logarithmic utility ($p=0$) in Section 5. Finally, Section 6 contains a counterexample illustrating the necessity of our assumptions.

\section{Main Result}

\subsection{Model and Main Result}
We focus on an agent who invests in assets $S$, thereby affecting total wealth $X$, so as to maximize expected utility from terminal wealth $X_T$ at time $T$:
\begin{equation}\label{eq: exp U}
\max_{X \text{ admissible}} \expec[U(X_T)] \ .
\end{equation}
To encompass the applications below, the utility function $U$ is strictly increasing and concave, but not necessarily continuously differentiable or strictly concave.
The model \eqref{eq: exp U} has the usual portfolio choice interpretation, in which the agent is an investor, $S$ represents the financial asset(s), $X$ the portfolio value, and $U$ the investor's utility function. A second interpretation is that of the agent as a corporate manager, $S$ a real investment opportunity, and $X$ the firm's value. In this case, the function $U$ combines the manager's preferences and incentives: for example, if the manager receives a compensation equal to $F(X_T)$ and has utility function $u$, then $U(x) = u (F(x))$.
A hybrid interpretation (cf. \citet*{Carpenter}) is that of a fund manager, who invests in financial assets $S$, so as to maximize the expected payoff of some function $U$ of the terminal fund value.

Formally, there are $d+1$ assets available. A safe asset, with price denoted by $S^0$, and $d$ risky assets, with prices $S=(S^1,\ldots,S^d)$. These assets are traded continuously, without frictions, and no arbitrage opportunities are available. Let $(\Omega,\mathcal F,(\mathcal F_t)_{t\ge 0},P)$ be a filtered probability space satisfying the usual conditions of right-continuity and saturatedness.

\begin{ass}[Assets]\label{ass: growth}
The safe asset $S^0: [0,\infty) \mapsto \Real$ is a deterministic, strictly positive function satisfying $S^0_t \uparrow \infty$ as $t\uparrow\infty$.\footnote{If the utility function is strictly concave and continuously differentiable, the safe asset can be stochastic as long as it is bounded from below and above by two deterministic processes $\underline{S}$ and $\overline{S}$ such that $\lim_{t\rightarrow \infty} \underline{S}_t =\infty$; cf. Remark~\ref{rem: stoch interest}.}

Moreover, the discounted prices $S/S^0$ of the risky assets are local martingales under some probability $Q$ equivalent to $P$.
\end{ass}

The above assumption $\lim_{t\rightarrow \infty} S^0_t = \infty$ is satisfied, for example, in models with positive interest rates bounded away from zero. The existence of a martingale measure ensures that the market is free of arbitrage opportunities (cf. \citet{DS98}).

The agent's objective is described by a utility function $U$, which incorporates the combined effect of preferences and incentives. Henceforth, fix $p<1$, and denote by $\tilde U(x)$ the isoelastic utility function, defined by $\tilde{U}(x)=x^p/p$, $0 \neq p<1$, resp.\ $\tilde{U}(x) = \log(x)$ for $p=0$. Compared to the benchmark $\tilde{U}$, the generic utility function $U$ satisfies the following restrictions, which yield the main result.

\begin{ass}[Utility]\label{ass: utility}
\noindent
\begin{enumerate}[i)]
\item
$U(x): (0,\infty)\rightarrow \Real$ is strictly increasing, concave, not necessarily differentiable or strictly concave for low wealth levels, but differentiable and strictly concave for large enough wealth.

\item
As wealth increases ($x\uparrow\infty$), the utility $U$ becomes similar to the isoelastic utility $\tilde{U}$, in that their marginal utilities are asymptotically equivalent:
 \begin{equation}\label{ass: conv}
 \lim_{x\uparrow \infty} \frac{U'(x)}{\tilde{U}'(x)} =1.
 \end{equation}

\item
The utility $U$ satisfies additional conditions at low wealth levels, depending on the sign of $p$ in $\tilde{U}(x) = x^p/p$\,:
\begin{enumerate}[a)]
 \item For $0<p<1$, $U$ is bounded from below;
 \item For $p=0$, i.e., $\tilde{U}(x)=\log(x)$,
 \begin{equation}\label{ass: U/tU'}
   \liminf_{x\downarrow 0} \frac{U(x)}{\tilde{U}'(x)} >-\infty.
\end{equation}
 \item For $p<0$, $\lim_{x\uparrow \infty} U(x) =0$ and \eqref{ass: U/tU'} is satisfied.
\end{enumerate}
\end{enumerate}
\end{ass}

Condition $i)$ implies that the agent is risk averse when wealth is high. Condition $ii)$ requires that, when the agent is rich, the utility (either by preferences or by incentives) is close to isoelastic, which is the central assumption in turnpike theorems. In particular, \eqref{ass: conv} implies that $U$ satisfies the Inada condition at infinity, i.e., $\lim_{x\uparrow \infty} U'(x) =0$. However, the Inada condition may not be satisfied at zero.

Condition $iii)$ is new, and requires that $U$ is not too risk averse compared to $\tilde U$ when wealth is low. For example, if $U(x) = x^{p^*}/p^*$ for $x$ small, where $p^*<1$, the condition in \eqref{ass: U/tU'} boils down to $p^*\geq p-1$, that is, the risk aversion $1-p^*$ of $U$ should not be greater than one plus the risk aversion $1-p$ of $\tilde U$ at low wealth. In general, \eqref{ass: U/tU'} means that the ratio of utilities $U(x)/\tilde{U}(x)$ does not diverge faster than $x^{-1}$ for  $p\neq 0$ or $(x\log x)^{-1}$ for $p=0$, as the wealth $x$ tends to zero. These conditions are satisfied, in particular, if the ratio between $U$ and $\tilde U$ remains bounded near zero. The example outlined in Section \ref{sec:example} and analyzed in Section \ref{sec:analysis} below shows that if \eqref{ass: U/tU'} is dropped, the main result can fail even in the Black-Scholes model.

The agent invests in the assets subject to the usual budget constraint: if $x$ denotes the initial capital, and $(\varphi^i_t)^{1\le i\le d}_{0\le t\le T}$ the number of shares of the $i$-th asset at time $t$, the corresponding wealth $X^\varphi_t$ equals
\[
X^\varphi_t = S^0_t \left( x + \int_0^t \varphi_s d(S_s/S^0_s) \right).
\]
To simplify notation, without loss of generality we set $x=1$, which amounts to scaling the numeraire by a factor of $x$. An $\Real^d$-valued process $\varphi$ is an admissible strategy if it is predictable, $S$-integrable, and the corresponding wealth process satisfies $X^\varphi_t \ge 0$ a.s.\ for all $t\ge 0$. The class of admissible wealth processes is denoted by
$$\mathcal{X}:= \{X: X_t\geq 0, \prob-a.s.\ \text{ for all } t\geq 0\}.$$
Thus, the value functions for the utility maximization problems for the generic utility $U$ and its isoelastic counterpart $\tilde U$ are:
\begin{equation}\label{eq: primary}
u^T(x) = \sup_{X\in \mathcal X}\expec[U(X)]
\qquad\text{and}\qquad
{\tilde u}^T(x) = \sup_{X\in \mathcal X}\expec[\tilde U(X)].
\end{equation}

The final assumption is that the isoelastic utility maximization problem is well posed. This assumption is necessary only for $p\ge 0$ because it is always satisfied if the utility function is bounded from above for $p<0$:

\begin{ass}[Isoelastic Wellposedness]\label{ass: wellposedness}
If $0\leq p<1$, let
$$
\tilde{V}(y) = \begin{cases} -y^q/q &\text{for $p \in (0,1)$ and $q=p/(p-1)$},\\ -\log(y)-1 &\text{for $p=0$},\end{cases}
$$
be the convex dual of the isoelastic utility $\tilde{U}$, and assume that
\begin{equation}\label{eq:dualfinite}
\inf_{Y\in \mathcal{Y}} \expec[\tilde{V}(Y_T)] <\infty,
\end{equation}
where $\mathcal{Y}$ is the set of \emph{stochastic discount factors}:
\begin{align*}
\mathcal{Y} := \{Y = \overline{Y}/S^0: \overline{Y}>0 \text{ with } \overline{Y}_0=1 \text{ such that } X\overline{Y} &\text{ is a supermartingale}\\
&\qquad \quad  \text{for all $X \in \mathcal{X}$}\}.
\end{align*}
\end{ass}

Condition \eqref{eq:dualfinite} ensures that the isoelastic dual (and in turn primal) problem is well posed. For $p \neq 0$, this requires the existence of the $q$-th moment for \emph{some} stochastic discount factor, which is satisfied, for example, if the asset price follows an It\^o process, and the market price of risk is bounded. Indeed, in this case one can choose the density process of the minimal martingale measure, and \eqref{eq:dualfinite} follows from Novikov's condition. The argument for $p=0$ is similar.

With the above notation and assumptions, the main result reads as follows.
\begin{thm}[Robustness of Isoelastic Portfolios]\label{thm:mt}
Let Assumptions \ref{ass: growth} - \ref{ass: wellposedness} hold. Then for any horizon $T>0$ there exist an optimal payoff $X^T_T$ for the generic utility $U$ and $\tilde{X}_T^T$ for the isoelastic utility $\tilde{U}$.\footnote{Here, the superscript $T$ indicates the optimal wealth process for horizon $T$, whereas the subscript $T$ refers to its evaluation at time $T$.}  They satisfy
\begin{equation}\label{eq: limit ce}
 \lim_{T\rightarrow \infty} \frac{U^{-1}(\expec[U(\tilde{X}^T_T)])}{U^{-1}(\expec[U(X^T_T)])}=1.
\end{equation}
That is, in the long run, the certainty equivalent of the isoelastic portfolio is arbitrarily close in relative terms to that of the optimizer.
\end{thm}

The strength of this result is that the certainty equivalents converge -- not just their growth rates. In contrast, the long-run portfolio choice literature based on risk-sensitive control\footnote{See  \citet*{MR1358100},
\citeauthor*{MR1675114} \citeyearpar{MR1675114,MR1790132},
and several others.
}
and large deviations\footnote{\citet*{MR1968944}, \citet*{follmer:aaa}} focuses on the maximization of the equivalent safe rate, defined as:
\[
\liminf_{T\rightarrow\infty} \frac1T \log U^{-1}(\expec[U(X_T)]).
\]
It is clear that if two processes $(X_t)_{t\ge 0}$ and $(\tilde X_t)_{t\ge 0}$ satisfy \eqref{eq: limit ce}, then they share the same equivalent safe rate (if it exists). It is also clear, however, that \eqref{eq: limit ce} is a much stronger property. As a trivial example, even with a risk-neutral utility ($U(x) = x$), consider the two processes $\tilde X_t = e^{(\mu-\sigma^2/2)t + \sigma W_t}$, $X_t = \frac12 e^{r t} + \frac12 e^{(\mu-\sigma^2/2)t + \sigma W_t}$, which correspond to a full stock investment, and to a half-stock, half-bond investment (without rebalancing) in a safe asset earning a constant interest rate $r$ and a stock following geometric Brownian motion with expected return $\mu>r$ and volatility $\sigma$. Then, both wealth processes have the same growth rate $\mu$, but the ratio of the corresponding certainty equivalents converges to 2. In fact, examples are also available, in which such a ratio diverges while the equivalent safe rate remains equal. Thus, even if two investment policies share the same equivalent safe rate, they may have very different certainty equivalents, which means that the agent may value one policy much higher than the other.

When \eqref{eq: limit ce} holds, however, the optimal policy can only be marginally better than its isoelastic counterpart for long horizons, in that the gain from choosing the superior policy is smaller than any fraction of the value of the inferior one. It is precisely this property that makes the theorem relevant for incentive schemes, and the next section explores in detail the theorem's implications in this area.

\subsection{Incentives}

\subsubsection*{Weakness of incentives with one exercise price}
If a manager has isoelastic preferences ($\tilde U(x)=x^p/p$ with $0\neq p<1$), and compensation that includes a cash component $c_1\geq 0$ and a fraction $c_2>0$ of the equity $X_t$, the objective function is
\begin{equation}\label{eq:private}
\expec[\tilde{U}(c_1 + c_2 X_T)],
\end{equation}
where $X$ runs through the class of admissible wealth processes. Suppose now that shareholders are concerned that the manager's high equity exposure is likely to discourage investment in projects with positive expected value, and are contemplating to grant $c_3>0$ call options with exercise price $K$, as an incentive to take risks. Such executive stock options typically have a vesting period of ten years, so that our focus on long horizons is relevant. Including the option grant, the manager maximizes the objective
\begin{equation}\label{eq:manager}
\expec[\tilde{U}(c_1+c_2 X_T +c_3 (X_T-K)^+)].
\end{equation}
Both the optimizer and the certainty equivalent of this problem are the same as for
\begin{equation}\label{eq:manager2}
\expec[{\bar U}(X_T)] \ ,
\end{equation}
where $\bar U(x) = \tilde{U}(c_1+c_2 x + c_3 (x-K)^+)/(c_2+c_3)^p$. In other words, awarding the option grant is equivalent to replacing the individual utility $\tilde U$ with the \emph{effective} utility $\bar U$.
This utility $\bar U$ is strictly increasing, differentiable on $(K,\infty)$, and satisfies $\bar U(\infty)=0$ if $p<0$ as well as
\begin{equation}\label{eq:lim1}
\lim_{x \uparrow \infty} \frac{\bar U'(x)}{\tilde{U}'(x)}=1,
\qquad
\lim_{x \downarrow 0} \frac{\bar U(x)}{\tilde{U}'(x)}=0.
\end{equation}
Thus, the effective utility $\bar U$ satisfies Assumption 2 for any fixed compensation $c_1\ge 0$, with the exception that $\bar U$ is no longer concave in a neighborhood of the exercise price $K$, which is illustrated in Figure \ref{fig:convex}. Indeed, creating such a region is the main purpose of option incentives.

\begin{figure}[tbp]
\includegraphics[width=0.8\textwidth]{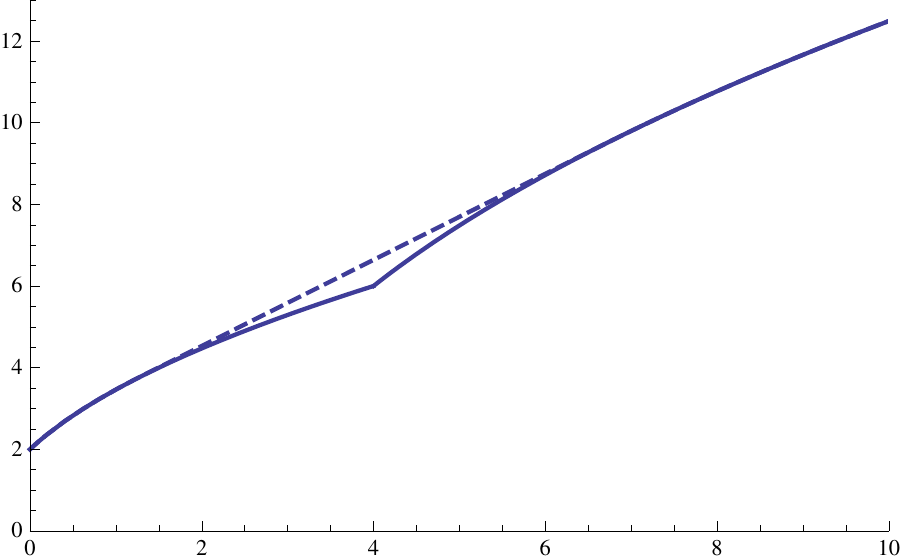}
\caption{Effective utility function $\bar{U}$ (solid) and corresponding concave envelope $U$ (dashed), for sample parameter values ($\gamma=1/2$, $c_1=1$, $c_2=2$, $c_3=3$, $K=4$).
}\label{fig:convex}
\end{figure}

Under different sets of assumptions\footnote{Which include a complete market where the unique equivalent martingale measure has no atoms, as well as other settings.} \citet{Carpenter}, \citet{cuoco2011equilibrium}, and \citet{Bichuch-Sturm} show that maximizing the expected utility $\bar U$ is actually equivalent to maximizing its concave envelope $U$, that is, the minimal concave function that dominates $\bar U$:
\begin{equation}\label{eq:manager3}
\expec[{U}(X_T)],
\qquad
\text{where }  U (x) = \inf \{ g(x) : g \text{ concave}, g\ge \bar U\} .
\end{equation}
In the present setting, the concave envelope coincides with $\bar U$ for sufficiently large or small wealth levels, therefore it preserves properties \eqref{eq:lim1} of $\bar U$. As it is also concave, Theorem \ref{thm:mt} applies to $U$, and yields the following result:

\begin{thm}\label{thm:incentive}
If the manager's problem \eqref{eq:manager} and the concavified problem \eqref{eq:manager3} are well posed and have the same solution, and the safe asset grows indefinitely, then, as the horizon increases, the certainty equivalents of the optimal strategies with incentives \eqref{eq:manager} and without incentives \eqref{eq:private} are asymptotically equivalent.
In particular, the manager's private value of an option grant with fixed exercise price becomes relatively negligible as the horizon increases.
\end{thm}

With hindsight, this result looks natural, and indeed follows from the the local effect on preferences of a single option. Yet, the option pricing intuition, which shapes much of the literature on incentives, points to the opposite conclusion. As in the Black-Scholes model (and also in more sophisticated extensions) the arbitrage-free price of a call option increases with volatility, it is plausible to conclude that a manager is encouraged to increase its value by taking more risk. A longer horizon also increases the option's value, suggesting a stronger, not weaker incentive effect.

However, option pricing heuristics are misleading in this context because the manager, who can neither sell nor hedge the option grant,\footnote{In practice, incentive stock options grants include clauses that prevent a manager from taking offsetting positions even with private accounts.} does not focus on the hypothetical risk-neutral value of the option, but rather on its private value, for which risk-aversion is central. As a single option affects risk-aversion only locally, long horizons combined with wealth growth make their impact vanish.

In summary, Theorem \ref{thm:incentive} supports the broad observation that the size and the exercise price of an option grant need to be chosen carefully, depending on the horizon.  A small quantity of options with a low strike may result, similar to a large stock position, in discouraging risk-taking. A large number of options with a high strike may also be ineffective, as the risk necessary to make the options profitable may be too much for the manager to bear. And even an exercise price chosen optimally at the time of award may soon become inadequate after large changes in the asset price.

\subsubsection*{Robustness of several strike prices, and power(ful) incentives}

We argue that option grants that include several (in theory, infinitely many) exercise prices retain their incentive effects after large price changes, and are robust to long horizons. Once again, the intuition comes from option theory: \citet{carr2001optimal} show that a European option with (smooth) payoff $f(S_T)$ admits the following representation as a portfolio of call and put options of all strikes:\footnote{The formula follows using the fundamental theorem of calculus twice, and then integrating by parts.}
\begin{equation}
\begin{split}
f(S_T) =\ &f(\bar K) + f'(\bar K)(S_T-\bar K)\\
&+\int_0^{\bar K} f''(k)(k-S_T)^+dk + \int_{\bar K}^\infty f''(k)(S_T-k)^+dk.
\end{split}
\end{equation}
In this representation, the term $f(\bar K)$ represents a cash amount, $f'(\bar K)(S_T-\bar K)$ a position in a forward contract, and the two integrals correspond to static portfolios in puts and calls, respectively. The threshold $\bar K$ is arbitrary, and determines the strike above which calls rather than puts are used.

For example, setting $\bar K=0$, this decomposition shows that a portfolio consisting of an equal number of call options at all strikes leads to a payoff of the form $f(x) = c x^2$ for some $c>0$. In general, a payoff of power type $f(x) = x^\alpha$ is replicated by a portfolio of options with weights $f''(x) = \alpha (\alpha -1) x^{\alpha - 2}$. We restrict to $\alpha>0$, as otherwise the incentive does not reward higher asset values.

Consider now the effect of such a power incentive on a manager with isoelastic utility $\tilde U(x) = x^p/p$. With fixed compensation set to zero for simplicity, the utility function including incentives becomes $U(x) = x^{\alpha p}/p$, which means that the incentive replaces the manager's risk aversion $\gamma = 1-p$ with the \emph{effective} risk aversion:\footnote{Here we assume that the resulting optimization problem is well-posed in the sense of Assumption \ref{ass: wellposedness}; in particular, the effective risk aversion should be positive.}
\begin{equation}\label{eq:gammastar}
\gamma^* = \alpha \gamma + (1-\alpha) .
\end{equation}
This formula has several implications. First, the incentive scheme reduces risk aversion ($\gamma^*<\gamma$) if and only if $(1-\alpha)(1-\gamma)<0$. In particular, a convex payoff ($\alpha>1$) does not reduce risk aversion if the latter is originally higher than logarithmic ($\gamma>1$), as it typically is. In general, a concave payoff ($\alpha<1$) makes a manager closer to logarithmic $\gamma=1$, and a convex payoff does the opposite.

Thus, for a manager who has an already low risk aversion ($\gamma<1$), an options grant that includes calls with several exercise prices is an incentive scheme that remains robust to the horizon, and to large movements in the stock price, preventing the need for future repricing.

The payoff of a concave incentive is qualitatively similar to that of a combination of covered-call positions, but devising concave incentives may not be practical, because they would imply large stock holdings, combined with short positions in options of all maturities. The unusual nature of such arrangements, combined with the resulting difficulties for tax and accounting purposes, may make such schemes hard to implement.

In spite of these difficulties, note that the above formula shows that such concave power incentive are implicit in the high-water mark provisions in the compensation of hedge fund managers. Indeed, \citet{guasoni2011incentives} find that a hedge fund manager with risk aversion $\gamma$, who receives as performance fee a fraction $1-\alpha$ of profits, invests the fund's assets like an owner-investor with the same effective risk aversion as in \eqref{eq:gammastar}. This observation shows that concave incentives, although virtually absent in corporate compensation, are in fact implicitly present in the hedge-fund industry, and that the typical performance fees of $20\%$ correspond to a power incentive with $\alpha=0.8$.

\subsection{Counterexample}\label{sec:example}

This section outlines a counterexample, which shows that Theorem \ref{thm:mt} can fail even in the usual Black-Scholes model if condition \eqref{ass: U/tU'} is not satisfied, i.e., if the generic utility $U$ is too risk averse compared to the reference isoelastic utility $\tilde U$ at low wealth levels. The detailed calculations are presented in Section \ref{sec:analysis} below.

Suppose the safe asset earns a constant interest rate $r>0$, and there is a single risky asset following geometric Brownian motion:
$$dS_t/S_t=(\mu+r)dt+\sigma dW_t,$$
for a standard Brownian motion $W_t$ and constants $\mu, \sigma>0$. Let $\tilde{U}(x)=x^p/p$, $p<0$ be a reference isoelastic utility with risk aversion $1-p>1$ and consider the generic utility function given by $U(x)=x^p/p$ for sufficiently large $x$ and by $U(x)=x^{p^*}/p^*$, $p^* < p-1$, for $x \leq 1$, with smooth interpolation in between. Then, at high wealth levels, the generic utility $U$ has the same risk aversion $1-p$ as its isoelastic counterpart $\tilde{U}$, but the corresponding risk aversion is bigger (by more than 1) at low wealth levels. In particular, a simple calculation shows that the generic utility $U$ does not satisfy \eqref{ass: U/tU'}.

In this setting, the optimal isoelastic portfolio $\widetilde{X}^T_T$ can be too risky for the generic utility at low wealth levels, leading to a diminishing ratio of certainty equivalents compared to the optimizer $X^T_T$ for $U$ in the long run:
\begin{equation}\label{eq: ce conv 0}
\lim_{T\rightarrow \infty} \frac{U^{-1}(\expec[U(\tilde{X}^T_T)])}{U^{-1}(\expec[U(X^T_T)])}=0.
\end{equation}
Indeed, we show in Section \ref{sec:analysis} that this result holds when the risky asset is attractive enough, relative to the safe asset:
\begin{equation}\label{eq: para rest}
\left(\frac{\mu}{\sigma^2 (1-p)}\right)^2 > 2\max \left\{1,\frac{1}{p-p^*-1}\right\} \frac r{\sigma^2}.
\end{equation}
To understand this parameter restriction, recall that in this model $\mu/(\sigma^2 (1-p))$ is the optimal portfolio weight in the risky asset for the isoelastic utility. Thus, the lower bound \eqref{eq: para rest} ensures that the risky asset is sufficiently attractive, compared to the safe asset, to lead to a sufficiently large risky investment. In particular, as the difference $p-p^*$ between risk aversions at low wealth levels declines to one, the isoelastic risky weight rises, leading to a sufficiently large probability of reaching low wealth levels. Indeed, the lower bound in \eqref{eq: para rest} tends to infinity, as the difference in risk aversions approaches one ($p^* \uparrow p-1$), in line with our main result, which holds if this difference is less than or equal to one.

The result in \eqref{eq: ce conv 0} may seem puzzling, because it implies that the properties of the generic utility function at low wealths are important in the long run, even though the safe asset has a positive growth rate, so that any initial safe investment grows arbitrarily, thereby avoiding low wealth with certainty. However, the optimal isoelastic portfolio keeps a constant \emph{fraction} of wealth in the risky asset. Hence, if the stock price drops and the risky position declines, the safe \emph{position} is reduced as well, so that wealth can decrease even further.\footnote{As an extreme case, recall that, if risk aversion is small enough, the wealth of the optimal Merton portfolio converges to zero almost surely, even though its expected return is high.} Such a portfolio may be unacceptable for another investor, who is substantially more risk averse at low wealth levels.

\section{Notation and Preliminaries}\label{sec:dual}
We begin by recalling the basic notions and duality results for the present non-smooth setting, which have been established by \cite{Bouchard-Touzi-Zeghal}.

Recall the value functions $u^T$ and $\tilde{u}^T$ for the generic utility $U$ and its isoelastic counterpart $\tilde{U}$ from \eqref{eq: primary}.
Let
$$V(y) := \sup_{x>0} (U(x) - xy) \quad \text{and} \quad\tilde{V}(y) := \sup_{x>0} (\tilde{U}(x) - xy)$$
be the dual functions of the generic utility $U$ and the isoelastic utility $\tilde{U}$, respectively. Define the domain of $V$ as $\text{dom}(V) := \{y>0: |V(y)|<\infty\}$ and let $\overline{\text{dom}(V)}$ be its closure. Consider the dual problems associated to \eqref{eq: primary}:
\begin{equation}\label{eq: dual}
 v^T(y) = \inf_{Y\in \mathcal{Y}} \expec[V(y Y_T)] \quad \text{ and } \quad \tilde{v}^T(y) = \inf_{Y\in \mathcal{Y}} \expec[\tilde{V}(yY_T)],
\end{equation}
where $\mathcal{Y}$ denotes the set of stochastic discount factors:
\begin{align*}
\mathcal{Y} := \{Y = \overline{Y}/S^0: \overline{Y}>0 \text{ with } \overline{Y}_0=1 \text{ such that } X\overline{Y} &\text{ is a supermartingale}\\
&\qquad \quad\text{for all $X \in \mathcal{X}$}\}.
\end{align*}

\begin{thm}[Bouchard-Touzi-Zeghal]\label{thm:duality BTZ}
 Suppose the following holds:
 \begin{enumerate}[a)]
  \item $U:(0,\infty)\rightarrow \Real$ is nonconstant, nondecreasing and concave;
  \item $\overline{\text{dom}(V)} = [0,\infty)$;
  \item $V$ satisfies the dual asymptotic elasticity condition
      \[
       AE_0(V) := \limsup_{y\downarrow 0} \sup_{x\in -\partial V(y)} \frac{yx}{V(y)}<\infty;
      \]
  \item There is $y>0$ such that $v^T(y)$ defined in \eqref{eq: dual} is finite.
 \end{enumerate}
 Then, optimal solutions $X^T \in \mathcal{X}$ for $u^T$ and $Y^T \in \mathcal{Y}$ for $v^T$ exist such that $X^T Y^T$ is a uniformly integrable martingale and
 \begin{equation}\label{eq: X^T resp}
  X^T_T \in -\partial V(y^T Y^T_T) \quad \text{for some $y^T>0$}.
 \end{equation}
\end{thm}

For isoelastic utilities $\tilde{U}$, the same statements hold (cf. \cite{Kramkov-Schachermayer-99}). We denote the corresponding optimal solutions of $\tilde{u}^T$ and $\tilde{v}^T$ by $\tilde{X}^T$ and $\tilde{Y}^T$, respectively. Both of them are unique. Moreover, because $\tilde{V}$ is differentiable, \eqref{eq: X^T resp} simplifies to $\tilde{X}^T_T = - \tilde{V}'(\tilde{y}^T \tilde{Y}^T_T)$ for some $\tilde{y}^T>0$.

\begin{rem}\label{rem: stoch interest}
If the generic utility $U$ is differentiable and strictly concave, the aforementioned duality results also hold for a nondeterministic safe asset $S^0$, assuming that the latter is bounded from above and below by two deterministic, positive processes $\underline{S}$ and $\overline{S}$ (cf. \cite[Theorem 3.10]{karatzas.zitkovic.03}).
\end{rem}

In what follows, we will verify all prerequisites of Theorem \ref{thm:duality BTZ} in our setting. We begin by deriving some consequences of the convergence \eqref{ass: conv} of the ratio of marginal utilities.

\begin{lem}\label{lem:rv}
Suppose \eqref{ass: conv} holds and let $p\neq 0$. Then, for large wealth levels, the generic utility $U$ lies between arbitrarily close multiples of its isoelastic counterpart $\tilde{U}$ up to an additive constant. That is, for any $\epsilon>0$ there exists some sufficiently large $M_\epsilon>0$ such that:
\begin{enumerate}
\item[i)] If $p\in (0,1)$, then for some constants $A_\epsilon$ and $B_\epsilon$,
 \begin{equation}\label{eq: U bdd p>0}
  (1-\epsilon) x^p/p + B_\epsilon \leq U(x) \leq (1+\epsilon) x^p/p + A_\epsilon, \quad \mbox{for $x\geq M_\epsilon$};
 \end{equation}

 \item[ii)] If $p<0$, then
 \begin{equation}\label{intest2}
 (1-\epsilon) x^p/p \geq U(x) \geq (1+\epsilon) x^p/p, \quad \mbox{for $x\geq M_\epsilon$.}
 \end{equation}
\end{enumerate}
Moreover, $U$ is regularly varying at infinity in both cases.
\end{lem}

\begin{proof}
For any $\epsilon\in (0,1)$, the convergence \eqref{ass: conv} of the ratio of marginal utilities yields the existence of some sufficiently large $M_\epsilon$ such that
 \begin{equation}\label{eq:intest}
 (1-\epsilon) x^{p-1} \leq U'(x) \leq (1+\epsilon) x^{p-1}, \quad \mbox{for $x\geq M_\epsilon$}.
 \end{equation}
Up to a larger $M_\epsilon$, we can assume that the generic utility $U$ is differentiable for $x \geq M_\epsilon$. When $p\in (0, 1)$, integrating \eqref{eq:intest} on $(M_\epsilon, x)$ for $x\geq M_\epsilon$ gives the estimates in i). When $p<0$, integrating \eqref{eq:intest} on $(x, \infty)$ for $x\geq M_\epsilon$ and using $U(\infty)=0$ in Assumption \ref{ass: utility} iii) c) yields the estimates in ii). As for the regular variation of $U$, the bounds in i) give
\[
 \limsup_{x\uparrow \infty} \frac{U(cx)}{U(x)} \leq \limsup_{x\uparrow \infty} \frac{(1+\epsilon) (cx)^p/p + A_\epsilon}{(1-\epsilon) x^p/p + B_\epsilon} = \frac{1+\epsilon}{1-\epsilon} \,c^p, \quad \mbox{for any $c>0$.}
\]
As $\epsilon$ is chosen arbitrarily, it follows that $\limsup_{x\rightarrow \infty} U(cx)/U(x) \leq c^p$. The converse statement $\liminf_{x\rightarrow \infty} U(cx)/U(x) \geq c^p$ is obtained analogously. Hence, $U$ is regularly varying at infinity. The regular variation of $U$ for the case $p<0$ is obtained along the same lines.
\end{proof}

As $U$ is differentiable and strictly concave at sufficiently large wealth levels (cf. Assumption \ref{ass: utility} i)), the subdifferential $-\partial V(y)$ is single-valued and equals $(U')^{-1}(y)$ when $y\in (0, y_0)$ for some sufficiently small $y_0>0$. Set
$$I(y):= -\partial V(y) = (U')^{-1}(y), \quad \text{for $y\in (0,y_0)$}.$$

\begin{lem}
Suppose \eqref{ass: conv} holds. Then:
\begin{equation}\label{eq: I conv}
 \lim_{y\downarrow 0} I(y) y^{\frac{1}{1-p}} =1.
\end{equation}
\end{lem}
\begin{proof}
Set $x= I(y)$, which tends to infinity as $y\downarrow 0$. Then, \eqref{eq: I conv} follows from
 \[
  \frac{I(y)}{y^{\frac{1}{p-1}}} = \frac{I(U'(x))}{(U'(x))^{\frac{1}{p-1}}} = \frac{x}{(U'(x))^{\frac{1}{p-1}}} = \pare{\frac{x^{p-1}}{U'(x)}}^{\frac{1}{p-1}} \rightarrow 1, \quad \text{as $y\downarrow 0$},
 \]
where the convergence holds due to \eqref{ass: conv}.
\end{proof}

Let us now verify the prerequisites of Theorem \ref{thm:duality BTZ}.

\begin{lem}\label{lem: dom V}
 Let Assumption \ref{ass: utility} hold. Then $\overline{\text{dom}(V)} = [0,\infty)$.
\end{lem}

\begin{proof}
 We prove the statement for the case $p<0$; the proofs for $0<p<1$ and $p=0$ are similar.
 As $U$ is increasing with $\lim_{x\uparrow \infty} U(x) =0$, we have $U(x) - xy \leq 0$ for any $x, y>0$. Therefore, $V(y)\leq 0$ for any $y>0$. To obtain a lower bound for $V$, recall the second inequality in \eqref{intest2}, which gives a lower bound for $U$ for $x\geq M_{\epsilon}$. For the same $M_\epsilon$, \eqref{ass: U/tU'} implies the existence of $C_M<0$ such that $U(x)\geq C_M x^{p-1}$ for $x< M_\epsilon$. Hence, when $y$ is small, $V(y)$ is bounded from below by the convex dual of $(1+\epsilon) x^p/p$, which is $-(1+\epsilon)^{\frac{1}{1-p}} \frac{p-1}{p} y^{\frac{p}{p-1}}$; when $y$ is large, $V(y)$ is bounded from below by the convex dual of $C_M x^{p-1}$, which is $-\frac{p-1}{p-2} ((p-1) C_M)^{-\frac{1}{p-2}} y^{\frac{p-1}{p-2}}$. Both convex dual functions are larger than negative infinity when $y\in (0,\infty)$. Therefore, $V(y)>-\infty$ for $y\in (0,\infty)$ confirming that $\text{dom}(V) =(0,\infty)$.
\end{proof}

\begin{lem}
The convex dual $V$ of the generic utility $U$ satisfies the dual asymptotic elasticity condition.
\end{lem}

\begin{proof}
As $-\partial V(y) = I(y)$ for sufficiently small $y$, the statement is equivalent to $\limsup_{y\downarrow 0} \frac{y I(y)}{V(y)}<\infty$. For any $\epsilon>0$, \eqref{eq: I conv} yields $(1-\epsilon) y^{\frac{1}{p-1}} \leq I(y) \leq (1+\epsilon) y^{\frac{1}{p-1}}$ for sufficiently small $y$, say $y\leq y_\epsilon$ for some $y_\epsilon$. Integrating the first inequality on $(y, y_\epsilon)$ yields $V(y) \geq (1-\epsilon) \tilde{V}(y) + D_\epsilon$, for $y\leq y_\epsilon$ and some $D_\epsilon$. Then, $\limsup_{y\downarrow 0} \frac{y I(y)}{V(y)}<\infty$ follows from this lower bound for $V(y)$ and the upper bound for $I(y)$.
\end{proof}

\begin{lem}\label{lem: finite}
The value functions $u^T, {\tilde u}^T$ in \eqref{eq: primary}, and $v^T, {\tilde v}^T$ in \eqref{eq:
dual} are finite.
\end{lem}

\begin{proof}
If $p<0$, \eqref{eq: I conv} yields $I(y)\leq (1+\epsilon)y^{\frac{1}{p-1}}$ for $y\leq y_\epsilon$. Integrating this inequality on $(y, y_\epsilon)$ yields $V(y)\leq (1+\epsilon) \tilde{V}(y) + \tilde{D}_\epsilon$, $y\leq y_\epsilon$, for some $\tilde{D}_\epsilon$. As $\tilde{V}(y)<0$ and $V(y)$ is decreasing, this upper bound implies that $V(y)$ is uniformly bounded from above on $(0,\infty)$, so that $v^T(y)<\infty$ for any $y>0$. On the other hand, it follows from $\tilde{U}(x) \leq 0$ and the first inequality in \eqref{intest2} that $U(x)$ is uniformly bounded from above. Hence, $u^T <\infty$ for any $T>0$ as well. When $0\leq p<1$, the same argument as in the case $p<0$ gives $V(y) \leq (1+\epsilon) \tilde{V}(y) + \tilde{D}_\epsilon$, $y\leq y_\epsilon$, for some $\tilde{D}_\epsilon$. As $V(y)$ is decreasing, $V(y)$ is also bounded from above for $y\geq y_\epsilon$. Combining this upper bound for $V(y)$ with Assumption \ref{ass: wellposedness} gives $v^T(y)<\infty$, for any $T>0$ and $y>0$. Hence $u^T$ is also finite, due to the duality relation $u^T \leq \inf_{y>0}(v^T(y) + y)$.
\end{proof}

The previous three lemmata verify all assumptions in Theorem \ref{thm:duality BTZ}, hence the statements therein hold. For the isoelastic problem $\tilde{u}^T$, we recall from \cite[Lemma 5]{Guasoni-Robertson} the following duality result, which will be frequently used below.

\begin{lem}\label{lem: power dual}
Let $X,Y$ be $\F_T$-measurable random variables such that $X,Y>0$ almost surely and $\mathbb{E}[XY] \leq 1$. Then, for any $0\neq p<1$,
$$\frac{1}{p} \mathbb{E}[X^p] \leq \frac{1}{p}\mathbb{E}[Y^q]^{1-p}, \quad \text{for $q=\frac{p}{p-1}$},$$
and equality holds if and only if $\mathbb{E}[XY]=1$ and $X^{p-1}=\alpha Y$ for some $\alpha>0$.
\end{lem}

Finally, we note that $\lim_{T\rightarrow \infty} S^0_T=\infty$ in Assumption \ref{ass: growth} implies the following long-run property for stochastic discount factors:
\begin{equation}\label{eq: expec Y conv}
 \limT \sup_{Y\in \mathcal{Y}}\expec[Y_T] =0.
\end{equation}
Indeed, a full safe investment is admissible, so that the supermartingale property of $S^0Y$ yields $S^0_0 \geq S^0_T\expec[Y_T]$ for any $Y\in \mathcal{Y}$. Assertion \eqref{eq: expec Y conv} in turn follows from Assumption~\ref{ass: growth}.

\section{Proof of the main result for $p\neq 0$}

Different arguments are needed to establish the main result in the cases where the corresponding isoelastic utilities are of power type ($p\neq 0$) or logarithmic ($p=0$), respectively.

First, consider the case $p \neq 0$.
Recall from Lemma \ref{lem:rv} that the generic utility $U$ is regularly varying at infinity. As a result, the convergence \eqref{eq: limit ce} of the ratio of certainty equivalents follows from the convergence of the ratio of the respective utilities:

\begin{lem}\label{lem:equiv}
Let Assumptions \ref{ass: growth} - \ref{ass: wellposedness} hold. Then \eqref{eq: limit ce} holds provided that
 \begin{equation}\label{eq: value conv 1}
  \limT \frac{\expec[U(X^T_T )]}{\expec[U(\tilde{X}^T_T )]} =1.
 \end{equation}
 \end{lem}

 \begin{proof}
To simplify notation, define $a_T = U^{-1}(\expec[U(X^T_T)])$ as well as $b_T = U^{-1}(\expec[U(\tilde{X}^T_T)])$, and note that $a_T, b_T\geq 0$ and $a_T\geq b_T$.

Observe that $\lim_{T\rightarrow \infty} a_T = \lim_{T\rightarrow \infty} b_T = \infty$. Indeed, for $p\in (0,1)$, $\limT \expec[U(X_T^T )] \geq \limT \expec[U(S^0_T )] = \infty$ follows from $\lim_{x\uparrow \infty} U(x) = \infty$ (cf. the first inequality in \eqref{eq: U bdd p>0}) and Assumption \ref{ass: growth}. In view of \eqref{eq: value conv 1}, we also have $\lim_{T\rightarrow \infty} \expec[U(\tilde{X}^T_T)] =\infty$. As $U^{-1}(\infty)=\infty$, this confirms that $\lim_{T\rightarrow \infty} a_T = \lim_{T\rightarrow \infty} b_T = \infty$. When $p<0$, the argument is analogous, with the difference that $U(\infty)=0$, whence both expected utilities converge to zero, and therefore certainty equivalents are also become infinite.

Now, prove \eqref{eq: limit ce} by contradiction. Suppose that, for any $\epsilon>0$, there exists a subsequence $\{T_n\}_{n\geq 0}$ with $\lim_{n\rightarrow \infty} T_n=\infty$ such that $a_{T_n} / b_{T_n} \geq 1+\epsilon$. For $p\in (0,1)$, the regular variation of $U$ at infinity implies that
 \[
  \liminf_{T_n\rightarrow \infty} \frac{U(a_{T_n})}{U(b_{T_n})} \geq \liminf_{T_n\rightarrow \infty} \frac{U((1+\epsilon) b_{T_n})}{U(b_{T_n})} = (1+\epsilon)^p >1.
 \]
For $p<0$, the inequality $U(b_{T_n}) <0$ and the regular variation of $U$ at infinity imply that
 \[
  \limsup_{T_n\rightarrow \infty} \frac{U(a_{T_n})}{U(b_{T_n})} \leq \limsup_{T_n\rightarrow \infty}\frac{U((1+\epsilon) b_{T_n})}{U(b_{T_n})} = (1+\epsilon)^p <1,
 \]
and both inequalities above contradict \eqref{eq: value conv 1}.
\end{proof}

To prove Theorem \ref{thm:mt}, it therefore remains to show the following:

\begin{prop}\label{prop: limit value}
Let Assumptions \ref{ass: growth} - \ref{ass: wellposedness} hold. Then
 \begin{equation}\label{eq: value conv}
  \limT \frac{\expec[U(X^T_T )]}{\expec[U(\tilde{X}^T_T )]} =1.
 \end{equation}
\end{prop}

In order to compare the utilities $\expec[U(X^T_T)]$ and $\expec[U(\tilde{X}^T_T)]$ as $T\rightarrow \infty$, compare both of them to the maximal isoelastic utilities $\expec[(\tilde{X}^T_T)^p/p]$. Henceforth, we prove Proposition \ref{prop: limit value} separately for $0<p<1$ and $p<0$, as the contribution of low wealth levels requires a different treatment in each case.

\subsection{Proof of Proposition \ref{prop: limit value} for $0<p<1$}

We first focus on the long-run performance of the isoelastic portfolio $\tilde{X}^T_T$.

\begin{lem}\label{lem: U-p-tx p>0}
 Let Assumptions \ref{ass: growth} - \ref{ass: wellposedness} hold. Then:
 \[
  \limT \frac{\expec[U(\tilde{X}^T_T)]}{\expec[(\tilde{X}^T_T)^p/p]} =1.
 \]
\end{lem}

\begin{proof}
 To simplify notation, the superscript $T$ in $\tilde{X}^T_T$ is omitted throughout this proof. By \eqref{eq: U bdd p>0},
  \begin{equation}
 \begin{split}
  &\frac{\expec\bra{U(\tilde{X}_T ) \, \indic_{\{\tilde{X}_T < M_\epsilon\}}}}{\expec[\tilde{X}^p_T/p]} + B_\epsilon \frac{\prob(\tilde{X}_T  \geq M_\epsilon)}{\expec[\tilde{X}^p_T/p]} + (1-\epsilon) \frac{\expec\bra{(\tilde{X}_T )^p \, \indic_{\{\tilde{X}_T \geq M_\epsilon\}}}}{\expec[\tilde{X}^p_T]} \\
  \leq\ & \frac{\expec[U(\tilde{X}_T )]}{\expec[\tilde{X}_T^p/p]} \label{eq: same X ub p>0}\\
  \leq\ & \frac{\expec\bra{U(\tilde{X}_T ) \, \indic_{\{\tilde{X}_T  < M_\epsilon\}}}}{\expec[\tilde{X}^p_T/p]} + A_\epsilon \frac{\prob(\tilde{X}_T  \geq M_\epsilon)}{\expec[\tilde{X}^p_T/p]} + (1+\epsilon) \frac{\expec\bra{(\tilde{X}_T )^p \, \indic_{\{\tilde{X}_T \geq M_\epsilon\}}}}{\expec[\tilde{X}^p_T]}.
 \end{split}
 \end{equation}
 Note that $U(\tilde{X}_T )$ is uniformly bounded from below because this is assumed for the generic utility $U$, cf. Assumption \ref{ass: utility} iii)-a). As $\limT\expec[\tilde{X}^p_T/p] \geq \limT\expec[(S^0_T)^p/p] =\infty$ by Assumption~\ref{ass: growth}, this implies that the first two terms on both sides of \eqref{eq: same X ub p>0} converge to zero as $T\rightarrow \infty$.

Now, focus on the third terms. Define an auxiliary probability measure $\prob^T$ on $\F_T$ via
 \begin{equation}\label{eq: def P^T}
  \frac{d\prob^T}{d\prob} = \frac{\tilde{X}^p_T}{\expec[\tilde{X}^p_T]} ,
 \end{equation}
where the expectation is finite by Lemma \ref{lem: finite}.
Recall from Section \ref{sec:dual} that $\tilde{y}^T \tilde{Y}^T_T = \tilde{X}_T^{p-1}$ and, moreover, $1=\expec[\tilde{Y}^T_T\tilde{X}_T]\geq \expec[\tilde{Y}^T_T X_T]$  and hence
 \begin{equation}\label{eq: numeraire}
  0\geq \expec\bra{\tilde{y}^T \tilde{Y}^T_T (X_T - \tilde{X}_T)} = \expec\bra{\tilde{X}^{p-1}_T (X_T - \tilde{X}_T)} = \expec\bra{\tilde{X}_T^p}\expec^{\prob^T}\bra{\frac{X_T}{\tilde{X}_T}-1},
 \end{equation}
 for any admissible wealth process $X$. This inequality yields that
 \begin{equation}\label{eq: tX ubb}
  \limT \prob^T(\tilde{X}_T \geq N) =1, \quad \text{ for any } N>0.
 \end{equation}
 Indeed, choosing $X_T=S^0_T$ in \eqref{eq: numeraire}, it follows that
 \[
  1\geq \expec^{\prob^T}[S^0_T/ \tilde{X}_T] \geq \expec^{\prob^T}\bra{S^0_T/\tilde{X}_T \, \indic_{\{S^0_T\geq L, \tilde{X}_T\leq N\}}} \geq \frac{L}{N} \, \prob^T(S^0_T\geq L, \tilde{X}_T\leq N),
 \]
 for any positive $L$ and $N$. Combined with Assumption \ref{ass: growth}, this yields
 \begin{align*}
  \limsup_{T\rightarrow \infty} \prob^T(\tilde{X}_T\leq N) &\leq \limsup_{T\rightarrow \infty} \prob^T(S^0_T\geq L, \tilde{X}_T \leq N) + \limsup_{T\rightarrow \infty} \prob^T(S^0_T <L)\\
  &\leq N/L,
 \end{align*}
 which confirms \eqref{eq: tX ubb} because $L$ was arbitrary.

 Now, coming back to the third term of the upper and lower bounds in \eqref{eq: same X ub p>0},
 \begin{align*}
  \frac{\expec\bra{(\tilde{X}_T )^p \, \indic_{\{\tilde{X}_T \geq M_\epsilon\}}}}{\expec[\tilde{X}^p_T]} =  \prob^T(\tilde{X}_T  \geq M_\epsilon) \rightarrow 1, \quad \text{as } T\rightarrow \infty,
 \end{align*}
 where the convergence follows from \eqref{eq: tX ubb}.
 Therefore, the above estimates on both sides of \eqref{eq: same X ub p>0} yield
 \[
   1-\epsilon \leq \liminf_{T\rightarrow \infty} \frac{\expec[U(\tilde{X}_T )]}{\expec[\tilde{X}^p_T]} \leq \limsup_{T\rightarrow \infty} \frac{\expec[U(\tilde{X}_T )]}{\expec[\tilde{X}^p_T]} \leq 1+\epsilon.
 \]
This proves the assertion because $\epsilon$ was chosen arbitrarily.
\end{proof}

The proof of Proposition \ref{prop: limit value} (and hence the main result) now proceeds as follows. The lower bound in \eqref{eq: U bdd p>0} yields
$$
\expec[U(\tilde{X}^T_T) \indic_{\{\tilde{X}^T_T \geq M_\epsilon\}}] \geq (1-\epsilon) \,\expec[(\tilde{X}^T_T)^p/p \, \indic_{\{\tilde{X}^T_T \geq M_\epsilon\}}] + B_\epsilon \prob(\tilde{X}^T_T \geq M_\epsilon),$$
which converges to $\infty$ as $T\rightarrow \infty$, because $\limT \expec[(\tilde{X}^T_T)^p/p] \geq \limT \expec[(S^0_T)^p/p] =\infty$ follows from Assumption \ref{ass: growth}. Therefore, given that $U$ is bounded from below, $\limT\expec[U(\tilde{X}^T_T)] =\infty$. On the other hand, the optimality of $X^T_T$ for the generic utility $U$ yields $\expec[U(\tilde{X}^T_T )] \leq \expec[U(X^T_T )]$, so that
\[
 \liminf_{T\rightarrow \infty}\frac{\expec[U(X^T_T )]}{\expec[U(\tilde{X}^T_T )]} \geq 1.
\]
Therefore, Proposition \ref{prop: limit value} follows by showing
\[
 \limsup_{T\rightarrow \infty} \frac{\expec[U(X^T_T )]}{\expec[U(\tilde{X}^T_T )]} \leq 1.
\]
In view of Lemma \ref{lem: U-p-tx p>0}, it suffices to prove that
\[
 \limsup_{T\rightarrow \infty} \frac{\expec[U(X^T_T)]}{\expec[(\tilde{X}^T_T)^p/p]} \leq 1.
\]
To establish this inequality, we begin with the following auxiliary result:

\begin{lem}\label{lem: ratio Y p>0}
Suppose Assumptions \ref{ass: growth} - \ref{ass: wellposedness} hold and let $p\in (0,1)$. Then:
\[
 1=\limT \frac{\expec[(Y^T_T)^q]^{1-p}}{\expec[(\tilde{Y}^T_T)^q]^{1-p}} = \limT \frac{y^T}{\tilde{y}^T}.
\]
\end{lem}

\begin{proof}
 To simplify notation, we omit the superscript $T$ in $X^T$, $\tilde{X}^T$, $Y^T$, and $\tilde{Y}^T$ throughout this proof.
 For the $M_\epsilon$ in \eqref{eq:intest}, the martingale property of $XY$ implies
 \begin{equation}\label{eq: XY primary p>0}
 \begin{split}
  1 = \expec[X_T Y_T] &= \frac{1}{y^T}\,\expec\bra{U'(X_T ) X_T \, \indic_{\{X_T \geq M_\epsilon\}}} + \expec[X_T Y_T \, \indic_{\{X_T <M_\epsilon\}}]\\
  &\leq \frac{1+\epsilon}{y^T} \,\expec\bra{X_T^{p-1} X_T \, \indic_{\{X_T  \geq M_\epsilon\}}} + \expec[X_T Y_T \, \indic_{\{X_T  <M_\epsilon\}}].
 \end{split}
 \end{equation}
 Here, the second identity uses the first-order condition $U'(X_T) = y^T Y_T$ when $X_T  \geq M_\epsilon$, and the inequality follows from the second inequality in \eqref{eq:intest}. The second term on the right-hand side of \eqref{eq: XY primary p>0} vanishes as $T\rightarrow \infty$ due to \eqref{eq: expec Y conv}. For the first term, note that $\expec[X_T \hat{Y}_T] \leq 1$ for any $\hat{Y}\in \mathcal{Y}$. Hence, it follows from Lemma \ref{lem: power dual} that
 \begin{equation}\label{eq: GR-duality}
  \frac{1}{p} \expec[X_T^p] \leq \frac{1}{p} \expec[\hat{Y}_T^q]^{1-p}, \quad \text{for $q=\frac{p}{p-1}$}.
 \end{equation}
 The previous inequality and $p>0$ imply that
 \[
  \liminf_{T\rightarrow \infty} \expec\bra{X_T^p \, \indic_{\{X_T  \geq M_\epsilon\}}} \leq \liminf_{T\rightarrow \infty} \expec\bra{X_T^p} \leq \liminf_{T\rightarrow \infty} \expec[\hat{Y}^q_T]^{1-p},\quad \mbox{ for any } \hat{Y}\in \mathcal{Y}.
 \]
 Coming back to \eqref{eq: XY primary p>0}, the above estimates for the two terms on the right-hand side yield
 \begin{equation}\label{eq: Y^q est p>0}
  \frac{1}{1+\epsilon} \leq \liminf_{T\rightarrow \infty} \frac{1}{y^T} \expec[\hat{Y}^q_T]^{1-p}, \quad \mbox{ for any } \hat{Y}\in \mathcal{Y}.
 \end{equation}

 For any $\epsilon>0$, \eqref{eq: I conv} shows that there exists a sufficiently small $\delta_\epsilon<y_0$, so that
 \begin{equation}\label{eq:estI p>0}
 (1-\epsilon) \, y^{\frac{1}{p-1}} \leq I(y) \leq (1+\epsilon) \, y^{\frac{1}{p-1}},  \quad  \mbox{for }  y<\delta_\epsilon.
 \end{equation}
Fix such a $\delta_\epsilon$; the martingale property of $X Y$ then implies
 \begin{equation}\label{eq: XY dual p>0}
  1= \expec[X_T Y_T] = \expec\bra{Y_T I(y^TY_T) \, \indic_{\{y^T Y_T\leq \delta_\epsilon\}}} + \expec\bra{Y_T X_T \, \indic_{\{y^T Y_T> \delta_\epsilon\}}},
 \end{equation}
 where the second identity follows from $X_T = I(y^TY_T)$ for $y^T Y_T \leq \delta_\epsilon$. Continue by estimating the second term on the right-hand side. Let $V'_-$ be the increasing left derivative of the convex function $V$. Then $-V'_-(y)$ dominates all other elements of the subdifferential $-\partial V$ at $y$, i.e., $- V'_-(y) \geq x$ for any $x\in -\partial V(y)$. Moreover, again as $-V'_-$ is nonincreasing, there exists $C_\delta$ such that $x\leq C_\delta$ for any $x\in -\partial V(y)$ and $y> \delta_\epsilon$. In view of  \eqref{eq: X^T resp},  it therefore follows that $X_T  \leq C_\delta$ when $y^T Y_T >\delta_\epsilon$. As a result, \eqref{eq: expec Y conv} gives
 \begin{equation*}\label{eq: YX large y}
  \expec\bra{Y_T X_T \, \indic_{\{y^T Y_T > \delta_\epsilon\}}} \leq C_\delta \expec\bra{Y_T \,\indic_{\{y^T Y_T >\delta_\epsilon\}}} \rightarrow 0, \quad \text{as} \quad T\rightarrow \infty.
 \end{equation*}
 Turning to the first term on the right-hand side of \eqref{eq: XY dual p>0}, it follows from the first inequality in \eqref{eq:estI p>0} that
 \begin{equation}\label{eq: YI est p>0}
 \begin{split}
  &\expec\bra{Y_T I(y^T Y_T) \, \indic_{\{y^T Y_T \leq \delta_\epsilon\}}}\\
  \geq\ & (1-\epsilon) \frac{\expec[Y_T^q \, \indic_{\{y^T Y_T\leq \delta_\epsilon\}}]}{(y^T)^{1/(1-p)}}\\
  =\ & \frac{1-\epsilon}{(y^T)^{1/(1-p)}} \expec[Y^q_T] - \frac{1-\epsilon}{(y^T)^{1/(1-p)}} \expec[Y^q_T \, \indic_{\{y^T Y_T >\delta_\epsilon\}}].
 \end{split}
 \end{equation}
 Here, the second term tends to zero as the horizon grows because
 \[
  \frac{1}{(y^T)^{1/(1-p)}} \expec[Y^q_T \, \indic_{\{y^T Y_T >\delta_\epsilon\}}] = \expec[Y_T (y^T Y_T)^{\frac{1}{p-1}} \, \indic_{\{y^T Y_T >\delta_\epsilon\}}] \leq \delta_\epsilon^{\frac{1}{p-1}} \expec[Y_T]\rightarrow 0,
 \]
 due to \eqref{eq: expec Y conv}.
 These estimates together with \eqref{eq: XY dual p>0} and \eqref{eq: YI est p>0} imply that
 \[
  \frac{1}{1-\epsilon} \geq \limsup_{T\rightarrow \infty} \frac{1}{(y^T)^{1/(1-p)}} \expec[Y^q_T].
 \]
Raising both sides to the power $(1-p)$, and using the optimality of $\tilde{Y}$ for the dual problem $\tilde{v}^T$ in \eqref{eq: dual}, we obtain
 \[
  \pare{\frac{1}{1-\epsilon}}^{1-p} \geq \limsup_{T\rightarrow \infty} \frac{1}{y^T} \expec[Y_T^q]^{1-p} \geq \limsup_{T\rightarrow \infty} \frac{1}{y^T} \expec[\tilde{Y}^q_T]^{1-p}.
 \]
Together with \eqref{eq: Y^q est p>0}, and recalling that $\epsilon$ was chosen arbitrarily, it follows that
 \[
  1= \limT \frac{1}{y^T} \expec[Y^q_T]^{1-p} = \limT \frac{1}{y^T} \expec[\tilde{Y}^q_T]^{1-p}.
 \]
Together with $\expec[\tilde{Y}_T^q]^{1-p} = \expec[\tilde{X}^p_T] = \tilde{y}^T$ (cf.\ Lemma \ref{lem: power dual}) this yields the assertion.
\end{proof}

We are now ready to complete the proof for Proposition \ref{prop: limit value}, and hence the main Theorem \ref{thm:mt}, for the case $0<p<1$.

\begin{proof}[Proof of Proposition \ref{prop: limit value} for $0<p<1$] We continue to omit the superscript $T$ to ease notation. As discussed before Lemma \ref{lem: ratio Y p>0}, it suffices to show
 \begin{equation}\label{eq: UX/ptX p>0}
  \limsup_{T\rightarrow \infty} \frac{\expec[U(X_T)]}{\expec[\tilde{X}_T^p/p]} \leq 1.
 \end{equation}
The second inequality in \eqref{eq: U bdd p>0} implies
 \begin{equation}\label{eq: est UX/ptX p>0}
 \begin{split}
  &\frac{\expec[U(X_T )]}{\expec[\tilde{X}_T^p/p]}\\
  =\ & \frac{\expec[U(X_T ) \,\indic_{\{X_T  <M_\epsilon\}}]}{\expec[\tilde{X}_T^p/p]} + \frac{\expec[U(X_T )\, \indic_{\{X_T  \geq M_\epsilon\}}]}{\expec[\tilde{X}_T^p/p]}\\
  \leq\ & \frac{\expec\bra{U(X_T ) \, \indic_{\{X_T  < M_\epsilon\}}}}{\expec[\tilde{X}^p_T/p]} + A_\epsilon \frac{\prob(X_T  \geq M_\epsilon)}{\expec[\tilde{X}^p_T/p]} + (1+\epsilon) \frac{\expec\bra{X_T^p \, \indic_{\{X_T  \geq M_\epsilon\}}}}{\expec[\tilde{X}^p_T]}\\
  \leq\ & \frac{\expec\bra{U(X_T ) \, \indic_{\{X_T  < M_\epsilon\}}}}{\expec[\tilde{X}^p_T/p]} + A_\epsilon \frac{\prob(X_T  \geq M_\epsilon)}{\expec[\tilde{X}^p_T/p]} + (1+\epsilon) \frac{\expec\bra{X_T^p}}{\expec[\tilde{X}^p_T]}.
 \end{split}
 \end{equation}
As we have seen in the proof of Lemma \ref{lem: U-p-tx p>0}, the first two terms on the right-hand side vanish as $T\rightarrow \infty$. For the third, Lemma \ref{lem: power dual} implies
 $
 (1/p)\expec[X_T^p] \leq (1/p) \expec[Y_T^q]^{1-p}
 $
so that $\limsup_{T\rightarrow \infty} \expec[X_T^p] \leq \limsup_{T\rightarrow \infty} \expec[Y_T^q]^{1-p}$ for $q=p/(p-1)$ because $p>0$. On the other hand, $\expec[\tilde{X}_T^p] = \expec[\tilde{Y}_T^q]^{1-p}$ (again cf. Lemma \ref{lem: power dual}). Therefore the previous estimates together with \eqref{eq: est UX/ptX p>0} imply
 \[
  \limsup_{T\rightarrow \infty} \frac{\expec[U(X_T )]}{\expec[\tilde{X}_T^p/p]} \leq (1+\epsilon) \limsup_{T\rightarrow \infty} \frac{\expec[Y_T^q]^{1-p}}{\expec[\tilde{Y}_T^q]^{1-p}} = 1+\epsilon,
 \]
 where the last identity follows from Lemma \ref{lem: ratio Y p>0}. Hence, Assertion \eqref{eq: UX/ptX p>0} is confirmed because $\epsilon$ was chosen arbitrarily.
\end{proof}

\subsection{Proof of Proposition \ref{prop: limit value} for $p<0$}

The overall strategy is similar to the case $0<p<1$. However, the contribution of low wealth levels to the total expected utility is more delicate, as utilities may be unbounded from below near zero and the value functions $u^T$ and $\tilde{u}^T$ converge to zero rather than infinity as $T\rightarrow \infty$.\footnote{This is because $0\geq u^T\geq \expec[U(S^0_T)]$ and $\limT \expec[U(S^0_T)]=0$ due to $\lim_{T\rightarrow \infty} S^0_T=\infty$ and $\lim_{x\uparrow \infty} U(x) =0$. The same statement holds for the isoelastic utility $\tilde{U}$.} Therefore, the additional Assumptions \eqref{ass: U/tU'} on $U$ are needed to ensure that the contribution of low wealth levels is still negligible in the long run. We start with the following analogue of Lemma \ref{lem: U-p-tx p>0}.

\begin{lem}\label{lem: U-p-tx}
Let Assumptions \ref{ass: growth} and \ref{ass: utility} hold. Then:
 \[
  \limT \frac{\expec[U(\tilde{X}^T_T)]}{\expec[(\tilde{X}^T_T)^p/p]} =1.
 \]
\end{lem}

\begin{proof}
 To simplify notation, the superscript $T$ in $\tilde{X}^T$ and $\tilde{Y}^T$ is omitted throughout this proof. In view of \eqref{intest2}, we have
 \begin{equation}\label{eq: same X ub}
 \begin{split}
  &\frac{\expec\bra{U(\tilde{X}_T ) \, \indic_{\{\tilde{X}_T < M_\epsilon\}}}}{\expec[\tilde{X}^p_T/p]} + (1-\epsilon) \frac{\expec\bra{\tilde{X}_T^p \, \indic_{\{\tilde{X}_T \geq M_\epsilon\}}}}{\expec[\tilde{X}^p_T]} \\
  \leq\ & \frac{\expec[U(\tilde{X}_T )]}{\expec[\tilde{X}_T^p/p]}\\
  \leq\ & \frac{\expec\bra{U(\tilde{X}_T ) \, \indic_{\{\tilde{X}_T  < M_\epsilon\}}}}{\expec[\tilde{X}^p_T/p]} + (1+\epsilon) \frac{\expec\bra{\tilde{X}_T^p \, \indic_{\{\tilde{X}_T \geq M_\epsilon\}}}}{\expec[\tilde{X}^p_T]}.
 \end{split}
 \end{equation}
 Let us estimate separately the two terms in the upper and lower bounds. For the second term, similar arguments as in the proof of Lemma \ref{lem: U-p-tx p>0}, utilizing the auxiliary measure $\prob^T$, yield
 \begin{equation}\label{eq: same X ub 2}
  \lim_{T\rightarrow \infty}  \frac{\expec\bra{\tilde{X}_T^p \, \indic_{\{\tilde{X}_T \geq M_\epsilon\}}}}{\expec[\tilde{X}^p_T]} = 1.
 \end{equation}

 Compared to the proof of Lemma \ref{lem: U-p-tx p>0}, the case $p<0$ differs with respect to the estimation of the first term in the upper and lower bounds in \eqref{eq: same X ub}. Here, the first-order condition $\tilde{X}_T^{p-1}=\tilde{y}^T \tilde{Y}_T$ and the martingale property of $\tilde{Y} \tilde{X}$ imply $\expec[\tilde{X}_T^p]=\tilde{y}^T$. On the other hand, it follows from \eqref{ass: U/tU'} that $pU(x) \leq C_M x^{p-1}$ for some constant $C_M>0$ and any $x<M_\epsilon$. As a result:
 \begin{equation}\label{eq: Utx/txp 2}
 \begin{split}
  \frac{\expec[U(\tilde{X}_T ) \, \indic_{\{\tilde{X}_T  <M_\epsilon\}}]}{\expec[\tilde{X}_T^p/p]} &=\frac{\expec[p\,U(\tilde{X}_T ) \, \indic_{\{\tilde{X}_T < M_\epsilon\}}]}{\tilde{y}^T} \\
  &\leq C_M \frac{\expec[\tilde{X}_T^{p-1} \, \indic_{\{\tilde{X}_T  <M_\epsilon\}}]}{\tilde{y}^T}\\
  &= C_M \expec[\tilde{Y}^T_T \indic_{\{\tilde{X}_T  < M_\epsilon\}}] \rightarrow 0, \quad \text{as $T\rightarrow \infty$}.
 \end{split}
 \end{equation}
 Therefore, the first terms of the upper and lower bounds in \eqref{eq: same X ub} vanish as the horizon grows. Together with \eqref{eq: same X ub} and \eqref{eq: same X ub 2}, it follows that
 \[
   1-\epsilon\leq \liminf_{T\rightarrow \infty} \frac{\expec[U(\tilde{X}_T)]}{\expec[\tilde{X}^p_T/p]} \leq \limsup_{T\rightarrow \infty} \frac{\expec[U(\tilde{X}_T )]}{\expec[\tilde{X}^p_T/p]} \leq 1+\epsilon.
 \]
This yields the assertion because $\varepsilon$ was arbitrary.
\end{proof}

\begin{rem}\label{rem: Utx/txp}
 The calculation of the contribution from low wealth levels relative to the expected power utility in \eqref{eq: Utx/txp 2} will be used in the counterexample in Section \ref{sec:analysis}. Therefore, for future reference, we summarize it here: for any $M$ and $U$ satisfying \eqref{ass: U/tU'},
 \[
  \lim_{T\rightarrow \infty} \frac{\expec[U(\tilde{X}^T_T) \, \indic_{\{\tilde{X}^T_T< M\}}]}{\expec[(\tilde{X}^T_T)^p/p]} =0.
 \]
 A careful examination of \eqref{eq: Utx/txp 2} shows that $U$ does not need to be concave to ensure the above convergence.
\end{rem}

Similarly as in the case $0<p<1$, the proof of Proposition \ref{prop: limit value} (and hence the main result) now proceeds as follows. The optimality of $X^T_T$ for the utility $U$ yields $\expec[U(\tilde{X}^T_T)] \leq \expec[U(X^T_T)] <0$, so that
\[
 \frac{\expec[U(X^T_T)]}{\expec[U(\tilde{X}^T_T)]} \leq 1, \quad \mbox{for any $T>0$.}
\]
Therefore, Proposition \ref{prop: limit value} follows by showing
\[
 \liminf_{T\rightarrow \infty} \frac{\expec[U(X^T_T)]}{\expec[U(\tilde{X}^T_T)]} \geq 1.
\]
Due to Lemma \ref{lem: U-p-tx}, it suffices to prove
\[
 \liminf_{T\rightarrow \infty} \frac{\expec[U(X^T_T)]}{\expec[(\tilde{X}^T_T)^p/p]} \geq 1,
\]
which will be established in the sequel. The following auxiliary result is the analogue of Lemma~\ref{lem: ratio Y p>0}:

\begin{lem}\label{lem: ratio Y}
Suppose Assumptions \ref{ass: growth} and \ref{ass: utility} hold for $p<0$. Then, even without assuming \eqref{ass: U/tU'}:
\[
 1=\limT \frac{\expec[(Y^T_T)^q]^{1-p}}{\expec[(\tilde{Y}^T_T)^q]^{1-p}} = \limT \frac{y^T}{\tilde{y}^T}.
\]
\end{lem}

\begin{proof}
 To simplify notation, we once more omit the superscript $T$ in $X^T$, $\tilde{X}^T$, $Y^T$, and $\tilde{Y}^T$ throughout this proof.
 Coming back to \eqref{eq: XY dual p>0} and using the second inequality in Equation \eqref{eq:estI p>0},
 \begin{equation}\label{eq: dual p<0 lb}
 \begin{split}
  1&\leq (1+\epsilon) \frac{\expec[Y_T^q \, \indic_{\{y^T Y_T\leq \delta_\epsilon\}}]}{(y^T)^{1/(1-p)}} + \expec\bra{Y_T X_T \, \indic_{\{y^T Y_T> \delta_\epsilon\}}}\\
  &\leq (1+\epsilon) \frac{\expec[Y_T^q]}{(y^T)^{1/(1-p)}} + \expec\bra{Y_T X_T \, \indic_{\{y^T Y_T> \delta_\epsilon\}}},
 \end{split}
 \end{equation}
 where $q=p/(p-1)$. Now, repeating the argument after \eqref{eq: XY dual p>0}, the second term on the right-hand side vanishes as $T\rightarrow \infty$. Then, raise both sides of \eqref{eq: dual p<0 lb} to the $(1-p)$-th power, obtaining
 \begin{equation}\label{eq: Y<tY}
  \pare{\frac{1}{1+\epsilon}}^{1-p} \leq \liminf_{T\rightarrow \infty} \frac{1}{y^T} \expec[Y^q_T]^{1-p} \leq \liminf_{T\rightarrow \infty} \frac{1}{y^T} \expec[\tilde{Y}^q_T]^{1-p},
 \end{equation}
 where the second inequality follows from the optimality of $\tilde{Y}$ for the dual problem $\tilde{v}^T$ in \eqref{eq: dual}.

 On the other hand, for any fixed $a>0$, it follows from \eqref{ass: conv} that
 \begin{equation}\label{eq: conv a}
  \lim_{x\uparrow \infty} \frac{U'(x)}{(x+a)^{p-1}} =1.
 \end{equation}
 Hence, for any $\epsilon>0$, there is $M_{\epsilon, a}>0$ such that $U'(x)$ exists and
 \begin{equation}\label{eq:intest a}
  (1-\epsilon)(x+a)^{p-1} \leq U'(x) \leq (1+\epsilon) (x+a)^{p-1}, \quad \mbox{for $x\geq M_{\epsilon, a}$}.
 \end{equation}
The martingale property of $XY$, the first-order condition $y^T Y_T=U'(X_T)$ for large $X_T$, and the first inequality in \eqref{eq:intest a} in turn yield
 \begin{equation*}\label{eq: XY primary}
 \begin{split}
  1&= \expec[X_T Y_T] \\
  &= \expec\bra{X_T Y_T \, \indic_{\{X_T < M_{\epsilon, a}\}}} - \frac{a}{y^T} \expec\bra{U'(X_T)  \, \indic_{\{X_T\geq M_{\epsilon, a}\}}}\\
  &\qquad \qquad \qquad \qquad \qquad \quad+ \frac{1}{y^T} \expec\bra{U'(X_T) (X_T +a) \, \indic_{\{X_T\geq M_{\epsilon, a}\}}}\\
  & \geq \expec\bra{X_T Y_T \, \indic_{\{X_T < M_{\epsilon, a}\}}} - a \expec\bra{Y_T  \, \indic_{\{X_T \geq M_{\epsilon, a}\}}}\\
  & \qquad \qquad \qquad \qquad \qquad \quad+ (1-\epsilon)\frac{1}{y^T}\expec\bra{(X_T+a)^{p} \, \indic_{\{X_T \geq M_{\epsilon, a}\}}}.
 \end{split}
 \end{equation*}
By \eqref{eq: expec Y conv}, the first two terms on the right-hand side go to zero as $T\rightarrow \infty$. Therefore
 \begin{equation}\label{eq: exp X/y}
 \begin{split}
  \frac{1}{1-\epsilon} &\geq \limsup_{T\rightarrow \infty} \frac{1}{y^T} \expec[(X_T+a)^p \, \indic_{\{X_T \geq M_{\epsilon, a}\}}] \\
  &\geq \limsup_{T\rightarrow \infty} \frac{1}{y^T} \expec[(X_T+a)^p] - \limsup_{T\rightarrow \infty} \frac{1}{y^T} \expec[(X_T+a)^p \, \indic_{\{X_T< M_{\epsilon, a}\}}].
 \end{split}
 \end{equation}
 Let us estimate the second term on the right-hand side below. To this end, $p<0$ and $X_T\geq 0$ imply that
 \begin{equation}\label{eq: p/y 1}
 \begin{split}
  \frac{1}{y^T} \expec[(X_T+ a)^p \, \indic_{\{X_T < M_{\epsilon, a}\}}] &\leq \frac{a^p}{y^T} \expec[\indic_{\{X_T < M_{\epsilon, a}\}}] \\
  &\leq \frac{a^p}{y^T} \expec\bra{\frac{y^TY_T}{\delta_M} \indic_{\{X_T < M_{\epsilon, a}\}}} \\
  &= \frac{a^p}{\delta_M} \expec[Y_T \indic_{\{X_T < M_{\epsilon, a}\}}] \rightarrow 0, \quad \text{as $T\rightarrow \infty$},
 \end{split}
 \end{equation}
 where the last step is once again a consequence of \eqref{eq: expec Y conv}. Here, $\delta_M$ is a positive constant such that $y^T Y_T \geq \delta_M$ on $\{X_T <M_{\epsilon, a}\}$. The reason for the existence of such a constant is the following. Notice that $X_T \in -\partial V(y^T Y_T)$ and that every element $x$ of $-\partial V(y)$ dominates $-V'_+(y)$, where $V'_+(y)$ is the right derivative of $V$ at $y$. As $V(y)$ is strictly convex when $y$ is close to zero, $-V'_+$ is strictly decreasing. Together with $-V'_+(0)=\infty$, this implies the existence of $\delta_M>0$ such that $-V'_+(y) \geq M_{\epsilon, a}$ for $y<\delta_M$, and in turn $y^T Y_T \geq \delta_M$ on $\{X_T  <M_{\epsilon, a}\}$.

In view of \eqref{eq: p/y 1}, the second term on the right-hand side of \eqref{eq: exp X/y} vanishes as the horizon grows, so that
 \begin{equation*}\label{eq: X/y ub}
  \frac{1}{1-\epsilon} \geq\limsup_{T\rightarrow \infty} \frac{1}{y^T} \expec[(X_T+ a)^p].
 \end{equation*}
Now, note that $\expec\bra{(X_T + a)\frac{\hat{Y}_T}{a\expec[\hat{Y}_T] + 1}}\leq 1$ for any $\hat{Y}\in \mathcal{Y}$. Hence, Lemma \ref{lem: power dual} shows
 \begin{equation}\label{eq: GR-duality-a}
  \frac{1}{p} \expec[(X_T+a)^p] \leq \frac{1}{p} \frac{\expec[\hat{Y}_T^q]^{1-p}}{(a\expec[ \hat{Y}_T]+1)^{-p}}, \quad \text{for $q=\frac{p}{p-1}$}.
 \end{equation}
 Taking into account that $\limT \expec[\hat{Y}_T]=0$ and $p<0$, the previous two inequalities imply
 \begin{equation}\label{eq: Y/y ub}
  \frac{1}{1-\epsilon} \geq \limsup_{T\rightarrow \infty} \frac{1}{y^T} \expec[(X_T+ a)^p] \geq \limsup_{T\rightarrow \infty} \frac{1}{y^T}\expec[\hat{Y}^q_T]^{1-p}, \quad \text{ for any } \hat{Y}\in \mathcal{Y}.
 \end{equation}
 In particular, this holds for $Y_T$ and $\tilde{Y}_T$. Combining this inequality with \eqref{eq: Y<tY} and recalling that $\epsilon$ was chosen arbitrarily, it follows that
 \[
  1= \limT \frac{1}{y^T} \expec[Y^q_T]^{1-p} = \limT \frac{1}{y^T} \expec[\tilde{Y}^q_T]^{1-p}.
 \]
Together with $\expec[\tilde{Y}_T^q]^{1-p} = \expec[\tilde{X}^p_T] = \tilde{y}^T$ (cf.\ Lemma \ref{lem: power dual}), this yields the assertion.
\end{proof}

We are now ready to complete the proof of Proposition \ref{prop: limit value} and hence Theorem \ref{thm:mt} in the case $p<0$.

\begin{proof}[Proof of Proposition \ref{prop: limit value} for $p<0$]
 As discussed before Lemma \ref{lem: ratio Y}, it suffices to show that
 \begin{equation}\label{eq: liminf U/p}
 \liminf_{T\rightarrow \infty} \frac{\expec[U(X_T)]}{\expec[\tilde{X}_T^p/p]} \geq 1.
\end{equation}
 Fix $a>0$. Recall that $\lim_{x\uparrow \infty} U'(x)/(x+a)^{p-1} =1$ by \eqref{eq: conv a}. As in \eqref{intest2}, it follows that
 \begin{equation}\label{eq: U bdd p<0 a}
  (1-\epsilon) (x+a)^p/p \geq U(x) \geq (1+\epsilon) (x+a)^p/p, \quad \mbox{for }  x\geq M_{\epsilon, a}.
 \end{equation}
Hence,
\begin{align*}
  \frac{\expec[U(X_T)]}{\expec[\tilde{X}_T^p/p]} &=  \frac{\expec[U(X_T)\,\indic_{\{X_T<M_{\epsilon, a}\}}]}{\expec[\tilde{X}_T^p/p]} + \frac{\expec[U(X_T)\,\indic_{\{X_T\geq M_{\epsilon, a}\}}]}{\expec[\tilde{X}_T^p/p]} \\
  &\geq (1-\epsilon) \frac{\expec[(X_T+a)^p \indic_{\{X_T
  \geq M_{\epsilon, a}\}}]}{\expec[\tilde{X}^p_T]}\\
  &= (1-\epsilon) \frac{\expec[(X_T+a)^p]}{\expec[\tilde{X}^p_T]} - (1-\epsilon) \frac{\expec[(X_T+a)^p \indic_{\{X_T
  < M_{\epsilon, a}\}}]}{\expec[\tilde{X}^p_T]},
\end{align*}
where the inequality follows from the first inequality in \eqref{eq: U bdd p<0 a}, multiplied by $p<0$ on both sides. We have seen in \eqref{eq: p/y 1} that
\[
 \limT \frac{1}{y^T}\expec[(X_T+a)^p \indic_{\{X_T
  < M_{\epsilon, a}\}}] =0.
\]
As $\expec[\tilde{X}^p_T] = \tilde{y}^T$ and $\limT y^T/\tilde{y}^T =1$ by Lemma \ref{lem: ratio Y}, it follows that
\[
 \limT \frac{\expec[(X_T+a)^p \indic_{\{X_T
  < M_{\epsilon, a}\}}]}{\expec[\tilde{X}^p_T]} =0,
\]
and in turn
\begin{equation}\label{eq: UX/ptX}
 \liminf_{T\rightarrow \infty} \frac{\expec[U(X_T)]}{\expec[\tilde{X}_T^p/p]} \geq (1-\epsilon) \liminf_{T\rightarrow \infty}\frac{\expec[(X_T+a)^p]}{\expec[\tilde{X}^p_T]}.
\end{equation}
Now take $\hat{Y}$ in \eqref{eq: GR-duality-a} to be $Y$. Then, as $\limT \expec[Y_T]=0$, and $\expec[\tilde{X}^p_T] = \expec[\tilde{Y}_T^q]^{1-p}$ by another application of Lemma~\ref{lem: power dual}, it follows from \eqref{eq: UX/ptX} and Lemma~\ref{lem: ratio Y} that
\[
\liminf_{T\rightarrow \infty} \frac{\expec[U(X_T)]}{\expec[\tilde{X}_T^p/p]} \geq (1-\epsilon) \liminf_{T\rightarrow \infty}\frac{\expec[(X_T+a)^p]}{\expec[\tilde{X}^p_T]}\geq (1-\epsilon) \liminf_{T\rightarrow \infty} \frac{\expec[Y^q_T]^{1-p}}{\expec[\tilde{Y}_T^q]^{1-p}} = 1-\epsilon.
\]
As $\epsilon$ was arbitrary, this proves \eqref{eq: liminf U/p} and in turn Proposition \ref{prop: limit value}.
\end{proof}

\section{Proof of the main result for $p=0$}

If $p=0$, the isoelastic utility $\tilde{U}$ is logarithmic and -- compared to the power case $p \neq 0$ -- different arguments are needed to establish the main result. In this case, the convergence \eqref{ass: conv} of the ratio of marginal utilities implies the following:

\begin{lem}\label{lem: inv U}
 Suppose \eqref{ass: conv} holds. Then, for any $\epsilon>0$, there exists a sufficiently large $M_\epsilon$ such that
 \[
  (1-\epsilon) (a-b) \leq \log \frac{U^{-1}(a)}{U^{-1}(b)} \leq (1+\epsilon) (a-b), \quad \mbox{for $a\geq b\geq M_\epsilon$}.
 \]
\end{lem}

\begin{proof}
 For any $\epsilon \in (0,1)$, the same argument as in Item $i)$ of Lemma \ref{lem:rv} yields the existence of $M_\epsilon$ such that
 \begin{equation}\label{eq: U bd log}
  (1-\epsilon) \log x + B_\epsilon \leq U(x) \leq (1+\epsilon) \log x +A_\epsilon, \quad \mbox{for } x\geq M_\epsilon
 \end{equation}
and constants $A_\epsilon, B_\epsilon$. The lower bound implies $\lim_{x\uparrow \infty} U(x) =\infty$. Set $y= U(x)$; then, $x\uparrow \infty$ as $y\uparrow \infty$. Combined with the convergence of ratio of marginal utilities \eqref{ass: conv}, this yields
 \[
  \frac{(U^{-1})'(y)}{U^{-1}(y)} = \frac{1}{x U'(x)} \rightarrow 1, \quad \mbox{as $y\uparrow \infty$}.
 \]
 Therefore,  $1-\epsilon \leq (\log U^{-1}(y))' \leq 1+\epsilon$ for $y\geq M_\epsilon$, after enlarging $M_\epsilon$ if necessary.  The assertion then follows from integrating these inequalities over the interval $(b, a)$ for $a\geq b\geq M_\epsilon$.
\end{proof}

In the long run, the generic expected utility diverges, both for the corresponding optimal portfolio and for its isoelastic counterpart:

\begin{lem}\label{lem: log value}
Let Assumptions \ref{ass: growth} - \ref{ass: wellposedness} hold. Then
\[
 \limT \expec[U(X^T_T)] = \limT \expec[U(\tilde{X}^T_T)] =\infty.
\]
\end{lem}

\begin{proof}
 Once again, the superscript $T$ for $X^T$, $\tilde{X}^T$, and $\tilde{Y}^T$ is omitted throughout this proof.
 As $\expec[U(X_T)]\geq \expec[U(\tilde{X}_T)]$ by optimality of $X_T$ for the generic utility $U$, the first convergence is implied by the second one. To prove the second convergence, first note that the lower bound in \eqref{eq: U bd log} yields
 \begin{equation}\label{eq: Utx log}
\lim_{T\rightarrow\infty}
\expec[U(\tilde{X}_T) \, \indic_{\{\tilde{X}_T \geq M_\epsilon\}}] \geq (1-\epsilon) \,\expec[\log(\tilde{X}_T) \,\indic_{\{\tilde{X}_T  \geq M_\epsilon\}}] + B_\epsilon \prob(\tilde{X}_T \geq M_\epsilon) = \infty
.
\end{equation}
Here, the last convergence holds because the optimality of $\tilde{X}_T$ for the logarithmic utility implies $\limT\expec[\log(\tilde{X}_T)] \geq \limT \expec[\log(S^0_T)] =\infty$ by $\lim_{T\rightarrow \infty} S^0_T = \infty$. Now, turn to the contribution of the low wealth levels $\{\tilde{X}_T  <M_\epsilon\}$. To this end, Assumption \eqref{ass: U/tU'} guarantees the existence of $C_M$ such that $U(x) \geq C_M x^{-1}$ for all $x<M_\epsilon$. Then,
 \[
  \expec[U(\tilde{X}_T ) \, \indic_{\{\tilde{X}_T  <M_\epsilon\}}]  \geq C_M \,\expec[(\tilde{X}_T)^{-1} \, \indic_{\{\tilde{X}_T  <M_\epsilon\}}].
 \]
Recall from Section \ref{sec:dual} that $\tilde{X}_T^{-1} = \tilde{y}^T\tilde{Y}_T$ and $\tilde{y}^T \equiv 1$.\footnote{The identity $\tilde{y}^T \equiv 1$ follows from $1= \expec[\tilde{X}_T/\tilde{X}_T] = \expec[\tilde{X}_T \tilde{y}^T \tilde{Y}_T] = \tilde{y}^T$, where $\expec[\tilde{X}_T \tilde{Y}_T] =1$ is used to obtain the third identity.}
 In view of \eqref{eq: expec Y conv}, the term on the right-hand side therefore converges to zero as $T\rightarrow \infty$, so that $\liminf_{T\rightarrow \infty}[U(\tilde{X}_T ) \, \indic_{\{\tilde{X}_T  < M_\epsilon\}}]\geq 0$. Combined with \eqref{eq: Utx log}, this confirms the second (and in turn the first) convergence in the assertion.
\end{proof}

Set $a_T = \expec[U(X^T_T)]$ and $b_T = \expec[U(\tilde{X}^T_T)]$. As we have seen above, $a_T \geq b_T$ and both utilities tend to infinity as $T\rightarrow \infty$. In view of Lemma \ref{lem: inv U}, the convergence $U^{-1}(a_T)/U^{-1}(b_T) \to 1$ of the ratio of corresponding certainty equivalents is thereby equivalent to the \emph{difference} $a_T-b_T$ of utilities vanishing in the long run. This is in contrast to the power case $p \neq 0$, where the convergence of the ratio of certainty equivalents was found to be equivalent to the convergence of the \emph{ratio} $a_T/b_T$ of utilities in Lemma \ref{lem:equiv}. As a result, different estimates are needed in the case $p=0$. More specifically, the proof of the main result is based on the following long-run asymptotics in this case, whose technical proof is deferred to Section \ref{sec:proofp0}.

\begin{prop}\label{prop: ratio X}
 Let $p=0$ and suppose Assumptions \ref{ass: growth} - \ref{ass: wellposedness} hold. Then:
 \begin{enumerate}
  \item[i)] $\prob-\limT \tilde{X}^T_T =\infty$;
  \item[ii)] $\prob-\limT X^T_T =\infty$;
  \item[iii)] $\prob-\limT X^T_T /\tilde{X}^T_T =1$.
 \end{enumerate}
 Here, $\prob-\limT$ denotes convergence in $\prob$-probability.
\end{prop}

With Proposition \ref{prop: ratio X} at hand, we can now complete the proof of the main result also in the case $p=0$:

\begin{proof}[Proof of Theorem \ref{thm:mt} for $p=0$]
 The superscript $T$ for $X^T$, $\tilde{X}^T$, and $\tilde{Y}^T$ is again omitted throughout this proof.
 As discussed above, due to Lemmas \ref{lem: inv U} and \ref{lem: log value}, it suffices to prove
 \[
  \limT \expec[U(X_T ) - U(\tilde{X}_T )] =0.
 \]
As $\expec[U(X_T)] \geq \expec[U(\tilde{X}_T )]$, we only need to show
 \begin{equation}\label{eq: Ux-Utx}
  \limsup_{T\rightarrow \infty} \expec[U(X_T ) - U(\tilde{X}_T )]\leq 0.
 \end{equation}
For any $\epsilon>0$, the convergence \eqref{ass: conv} of the ratio of marginal utilities implies that there exists a sufficiently large $M_\epsilon$ such that $x U'(x) \leq 1+\epsilon$ for $x\geq M_\epsilon$. As a result:
 \[
  U(x) - U(\tilde{x}) = \int_{\tilde{x}}^x \frac{1}{y} y U'(y)\, dy \leq (1+\epsilon) (\log x - \log\tilde{x}), \quad \mbox{for } x\geq \tilde{x}\geq M_\epsilon.
 \]
Choosing $x= X_T$ and $\tilde{x}= \wt{X}_T$ in the previous inequality leads to
 \begin{equation}\label{eq: diff U est>M}
 \begin{split}
  &\expec\bra{\pare{U(X_T) - U(\tilde{X}_T)} \,\indic_{\{\tilde{X}_T \geq M_\epsilon, X_T  \geq M_\epsilon\}}} \\
   \leq\  &\expec\bra{\pare{U(X_T) - U(\tilde{X}_T)} \,\indic_{\{X_T  \geq \tilde{X}_T \geq M_\epsilon\}}}\\
  \leq\  &(1+\epsilon)\,\expec\bra{\log \frac{X_T}{\tilde{X}_T} \,\indic_{\{X_T \geq \tilde{X}_T\geq M_\epsilon\}}},
 \end{split}
 \end{equation}
 where the first inequality holds because $U(X_T ) < U(\tilde{X}_T )$ when $X_T  < \tilde{X}_T $.
 Now, observe that
 \[
  0 \leq \expec\bra{\log \frac{X_T }{\tilde{X}_T } \, \indic_{\{X_T  > \tilde{X}_T  \geq M_\epsilon\}}} \leq \expec\bra{\log \frac{X_T }{\tilde{X}_T } \, \indic_{\{X_T > \tilde{X}_T\}}} = \expec\bra{\log \pare{\frac{X_T }{\tilde{X}_T } \vee 1}},
 \]
 and, moreover,
 \begin{align*}
  \expec\bra{\exp \log \pare{\frac{X_T }{\tilde{X}_T } \vee 1}} = \expec\bra{\frac{X_T }{\tilde{X}_T } \vee 1} \leq \expec\bra{\frac{X_T }{\tilde{X}_T }} + 1 \leq 2, \quad \mbox{for all } T>0,
 \end{align*}
 where the last inequality holds due to the numeraire property of $\tilde{X}_T$, that is, $\expec[\hat{X}_T/\tilde{X}_T] \leq 1$ for any admissible $\hat{X}$ (cf.\ \eqref{eq: numeraire} with $p=0$). As a result, de la Vallee-Poussin's criterion as in \cite[Lemma 3]{Shiryaev} implies that the family $\log\pare{\frac{X_T}{\tilde{X}_T } \vee 1}$ is uniformly integrable in $T$. Together with Proposition \ref{prop: ratio X} iii), this  yields
 \begin{equation}\label{eq: lim log ratio}
  \limT \expec\bra{\log \frac{X_T }{\tilde{X}_T } \, \indic_{\{X_T  > \tilde{X}_T  \geq M_\epsilon\}}} =0.
 \end{equation}
In view of \eqref{eq: diff U est>M}, it follows that
 \[
  \limsup_{T\rightarrow \infty} \expec\bra{\pare{U(X_T ) - U(\tilde{X}_T)} \, \indic_{\{\tilde{X}_T  \geq M_\epsilon, X_T  \geq M_\epsilon\}}} \leq 0.
 \]
 In the next two paragraphs, we will show
 \begin{align}
  & \limsup_{T\rightarrow \infty} \expec\bra{U(X_T) \,\indic_{\{X_T \leq M_\epsilon \text{ or } \wt{X}_T  \leq M_\epsilon\}}} \leq 0, \label{eq: diff U est 1}\\
  & \liminf_{T\rightarrow \infty} \expec\bra{U(\tilde{X}_T)\,\indic_{\{X_T  \leq M_\epsilon \text{ or } \tilde{X}_T  \leq M_\epsilon\}}} \geq 0. \label{eq: diff U est 2}
 \end{align}
Combining these three inequalities then yields \eqref{eq: Ux-Utx}, and in turn completes the proof of Theorem~\ref{thm:mt} in the case $p=0$.

 To establish \eqref{eq: diff U est 1}, note that Proposition \ref{prop: ratio X} ii) gives
 $$\limsup_{T\rightarrow \infty} \expec\bra{U(X_T) \,\indic_{\{X_T  \leq M_\epsilon\}}}\leq U(M_\epsilon) \limsup_{T\rightarrow \infty} \prob(X_T  \leq M_\epsilon) =0.$$
 On the other hand, the upper bound on the generic utility $U$ in \eqref{eq: U bd log} implies
 \begin{align*}
  & \expec[U(X_T ) \, \indic_{\{X_T  > M_\epsilon, \tilde{X}_T  \leq M_\epsilon\}}]\\
  \leq\ & (1+\epsilon) \,\expec[\log(X_T ) \, \indic_{\{X_T  > M_\epsilon, \tilde{X}_T  \leq M_\epsilon\}}] + A_\epsilon \prob(\tilde{X}_T  \leq M_\epsilon)\\
  \leq\ & (1+\epsilon)\, \expec\bra{\log \frac{X_T }{\tilde{X}_T } \, \indic_{\{X_T > \tilde{X}_T\}}} + (1+\epsilon)\,\expec[\log(\tilde{X}_T ) \, \indic_{\{X \geq M_\epsilon, \tilde{X}_T  \leq\ M_\epsilon\}}]\\
  &\qquad  + A_\epsilon \prob(\tilde{X}_T  \leq M_\epsilon)\\
  \leq\ & (1+\epsilon)\, \expec\bra{\log \frac{X_T }{\tilde{X}_T } \, \indic_{\{X_T > \tilde{X}_T\}}} + ((1+\epsilon) \log M_\epsilon + A_\epsilon) \,\prob(\tilde{X}_T  \leq M_\epsilon)\\
  =\ & (1+\epsilon) \expec\bra{\log \frac{X_T }{\tilde{X}_T } \vee 1} + ((1+\epsilon) \log M_\epsilon + A_\epsilon) \,\prob(\tilde{X}_T  \leq M_\epsilon).
 \end{align*}
Now, Proposition \ref{prop: ratio X} iii) and the uniform integrability established in the derivation of \eqref{eq: lim log ratio} show that the first term on the right-hand side converges to zero as $T \to \infty$. Likewise, by Proposition \ref{prop: ratio X} i), the second term also tends to zero as the horizon grows, confirming \eqref{eq: diff U est 1}.

 To prove \eqref{eq: diff U est 2}, note that
 $$\liminf_{T\rightarrow \infty} \expec[U(\tilde{X}_T ) \, \indic_{\{\tilde{X}_T  >M_\epsilon, X_T  \leq M_\epsilon\}}] \geq U(M_\epsilon) \limT \prob(\tilde{X}_T  > M_\epsilon, X_T  \leq M_\epsilon) =0$$
 by Proposition \ref{prop: ratio X} ii). It therefore suffices to show \begin{equation}\label{eq: Utx log <M}
 \liminf_{T\rightarrow \infty} \expec[U(\tilde{X}_T ) \, \indic_{\{\tilde{X}_T  \leq M_\epsilon\}}]\geq 0.
 \end{equation}
 To this end, \eqref{ass: U/tU'} yields the existence of $C_M$ such that $U(x) \geq C_M x^{-1}$ for $x\leq M_\epsilon$. Together with the first-order condition $\tilde{X}_T^{-1}=\tilde{y}^T\tilde{Y}_T$ with $\tilde{y}^T\equiv 1$ and \eqref{eq: expec Y conv}, it follows that
 \begin{align*}
  \expec[U(\tilde{X}_T ) \, \indic_{\{\tilde{X}_T  \leq M_\epsilon\}}] \geq C_M \, \expec[ (\tilde{X}_T)^{-1} \, \indic_{\{\tilde{X}_T  \leq M_\epsilon\}}] = C_M \expec[\tilde{Y}_T \, \indic_{\{\tilde{X}_T  \leq M_\epsilon\}}] \rightarrow 0,
 \end{align*}
 as $T\rightarrow \infty$. Therefore \eqref{eq: Utx log <M} and in turn  \eqref{eq: diff U est 2} is confirmed in both cases, completing the proof of Theorem~\ref{thm:mt}.
\end{proof}

\subsection{Proof of Proposition \ref{prop: ratio X}}\label{sec:proofp0}

This section concludes the proof of Theorem \ref{thm:mt} in the case $p=0$, by establishing the auxiliary Proposition \ref{prop: ratio X}. In all the following proofs, the superscript $T$ in $\tilde{X}^T$, $X^T$, $Y^T$, and $\tilde{Y}^T$ is omitted to ease notation. We begin with the analogue of Lemma \ref{lem: ratio Y p>0}:

\begin{lem}\label{eq: ratio y p=0}
 Let $p=0$ and suppose Assumptions \ref{ass: growth} - \ref{ass: wellposedness} hold. Then
 \[
  \lim_{T\rightarrow \infty} \frac{y^T}{\tilde{y}^T} =1,
 \]
 where $\tilde{y}^T\equiv 1$ as we have seen before (cf.\ Footnote 15).
\end{lem}

\begin{proof}
 The proof follows the argument in Lemma \ref{lem: ratio Y p>0}, where many estimates are simplified when $p=0$. Indeed, on the one hand, the same estimate as in \eqref{eq: XY primary p>0} yields
 \[
  1 \leq \frac{1+\epsilon}{y^T} \,\expec\bra{\frac{X_T}{X_T} \, \indic_{\{X_T  \geq M_\epsilon\}}} + \expec[X_T Y_T \, \indic_{\{X_T  <M_\epsilon\}}] \leq \frac{1+\epsilon}{y^T} + \expec[X_T Y_T \, \indic_{\{X_T  <M_\epsilon\}}].
 \]
Here, the term $\expec[X_T Y_T \, \indic_{\{X_T  <M_\epsilon\}}]$ vanishes as $T\rightarrow \infty$ due to \eqref{eq: expec Y conv}, so that
 \begin{equation}\label{eq: ratio y p=0 liminf}
  \frac{1}{1+\epsilon} \leq \liminf_{T\rightarrow \infty} \frac{1}{y^T}.
 \end{equation}
 On the other hand, the same argument that leads to \eqref{eq: XY dual p>0} also yields
 \[
  1= \expec\bra{Y_T I(y^TY_T)\, \indic_{\{y^T Y_T\leq \delta_\epsilon\}}} + \expec\bra{Y_T X_T \, \indic_{\{y^T Y_T> \delta_\epsilon\}}},
 \]
 where the second term on the right-hand side vanishes as $T\rightarrow \infty$ by the same reasoning as after \eqref{eq: XY dual p>0}. Concerning the first term, \eqref{eq: I conv} for $p=0$ shows that $I(y)\geq (1-\epsilon)y^{-1}$ for any $y<\delta_\epsilon$ and some sufficiently small $\delta_\epsilon$. Therefore:
 \[
 \expec\bra{Y_T I(y^T Y_T) \, \indic_{\{y^T Y_T \leq \delta_\epsilon\}}}
  \geq (1-\epsilon) \frac{\prob(y^T Y_T\leq \delta_\epsilon)}{y^T}
  = \frac{1-\epsilon}{y^T} - \frac{1-\epsilon}{y^T} \prob(y^T Y_T >\delta_\epsilon).
 \]
 Here the second term also tends to zero because
 \[
  \frac{1}{y^T}\prob(y^TY_T >\delta_\epsilon) = \expec\bra{\frac{Y_T}{y^T Y_T} \indic_{\{y^T Y_T >\delta_\epsilon\}}} \leq \delta_{\epsilon}^{-1} \expec[Y_T]\rightarrow 0,
 \]
 due to \eqref{eq: expec Y conv}. Combining the above estimates, we obtain
 \begin{equation}\label{eq: ratio y p=0 limsup}
  \frac{1}{1-\epsilon}\geq \limsup_{T\rightarrow \infty} \frac{1}{y^T}.
 \end{equation}
 Thus, the assertion follows by combining \eqref{eq: ratio y p=0 liminf} and \eqref{eq: ratio y p=0 limsup} because $\varepsilon$ was arbitrary.
\end{proof}

Using the previous results, we can now verify the first two items of Proposition \ref{prop: ratio X}.

\begin{lem}\label{lem: X tX ub}
Let $p=0$ and suppose Assumptions \ref{ass: growth} - \ref{ass: wellposedness} hold. Then:
 \[
  \prob-\limT \tilde{X}^T_T = \infty \quad \mbox{and} \quad \prob-\limT X^T_T = \infty.
 \]
\end{lem}
\begin{proof}
Recall the numeraire property \eqref{eq: numeraire} of the log-optimal portfolio $\tilde{X}$: $\expec[\hat{X}_T / \tilde{X}_T] \leq 1$ for any $T$ and any admissible payoff $\hat{X}$. The first part of the assertion in turn follows verbatim as in the derivation of \eqref{eq: tX ubb}, where $\prob^T$ is replaced by $\prob$. As for the second part of the assertion, for any $N>0$, there exists $y_N$ such that $I(y) > N$ for $y\leq y_N$. Here, $y_N$ is chosen sufficiently small so that $-\partial V(y)$ is single valued and is denoted by $I(y)$ for $y\leq y_N$. For the chosen $N$,
 \begin{align*}
  \prob(X_T \leq N) \leq \prob(X_T  \leq N, y^T Y_T \leq y_N) + \prob(y^T Y_T >y_N).
 \end{align*}
 The second term on the right-hand side vanishes as $T\rightarrow \infty$, due to \eqref{eq: expec Y conv}, $\limT y^T/\tilde{y}^T = 1$, and $\tilde{y}^T \equiv 1$ (cf. Lemma \ref{eq: ratio y p=0}). The first term is identically zero, because $X_T  = I(y^T Y_T) >N$ on $\{y^T Y_T \leq y_N\}$. Therefore, the second part of the assertion follows.
\end{proof}

To complete the proof of Proposition  \ref{prop: ratio X}, it remains to verify Item iii). To this end, some auxiliary results are established first.

\begin{lem}\label{lem: lim XY}
Suppose Assumptions \ref{ass: growth} - \ref{ass: wellposedness} hold for $p=0$. Then:
\[
 \prob-\limT X^T_T Y^T_T =1.
\]
\end{lem}
\begin{proof}
 For any $\epsilon>0$, the convergence of the ratio of marginal utilities \eqref{ass: conv} implies the existence of $M_\epsilon$ such that $U(x)$ is differentiable beyond $M_\epsilon$ and, moreover, $xU'(x) \geq 1-\epsilon/2$  for $x\geq M_\epsilon$. For such a $M_\epsilon$,
 \begin{equation}\label{eq: p XY<1-e}
  \prob(X_TY_T \leq 1-\epsilon) \leq \prob(X_T Y_T \leq 1-\epsilon, X_T  \geq M_\epsilon) + \prob(X_T  < M_\epsilon).
 \end{equation}
 In view of Lemma \ref{lem: X tX ub}, the second term on the right-hand side vanishes as $T\rightarrow \infty$. To estimate the first term, recall that $Y_T = U'(X_T ) /y^T$ on $\{X_T  \geq M_\epsilon\}$. As a result, for $T \geq T_*$,
 \begin{align*}
  \prob(X_T Y_T \leq 1-\epsilon, X_T  \geq M_\epsilon) &= \prob\pare{X_T  U'(X_T ) \leq y^T (1-\epsilon), X_T  \geq M_\epsilon}\\
  &\leq \prob(1-\epsilon/2\leq y^T (1-\epsilon)),
 \end{align*}
 where the inequality holds because $X_T U'(X_T ) \geq 1-\epsilon/2$ on $\{X_T  \geq M_\epsilon\}$. As $\limT y^T=1$ by Lemma \ref{eq: ratio y p=0}, the above estimates imply
 $$\limT \prob(X_T Y_T \leq 1-\epsilon, X_T  \geq M_\epsilon)=0.$$
Coming back to \eqref{eq: p XY<1-e}, this gives $\limT \prob(X_T Y_T \leq 1-\epsilon)=0$. Along the same lines, $\limT \prob(X_T Y_T \geq 1+\epsilon) =0$ follows, completing the proof.
\end{proof}

\begin{cor}\label{cor: XY conv}
Let $p=0$ and suppose Assumptions \ref{ass: growth} - \ref{ass: wellposedness} hold. Then:
\[
 \limT \expec[|X^T_T Y^T_T-1|] =0.
\]
\end{cor}

\begin{proof}
By the convergence in probability established in the previous lemma and $\expec[X^T_TY^T_T]=1$, this follows from Scheffe's lemma \cite[5.5.10]{williams.91}.
\end{proof}

Set $r_T = X^T_T /\tilde{X}^T_T$. The following estimate is a key to prove Proposition \ref{prop: ratio X} iii).

\begin{lem}\label{lem: r est}
Let $p=0$ and suppose Assumptions \ref{ass: growth} - \ref{ass: wellposedness} hold. Then:
 \[
  \limT \expec\bra{\left|1-\frac{X^T_T Y^T_T}{r_T}\right| |r_T-1|}=0.
 \]
\end{lem}

\begin{proof}
 Recall the first-order condition $\tilde{y}^T \tilde{Y}_T = \tilde{X}^{-1}_T$. As $1= \expec[\tilde{Y}_T \tilde{X}_T] \geq \expec[\tilde{Y}_T X_T]$, it follows that
 \[
  \expec[\tilde{X}^{-1}_T (X_T - \tilde{X}_T)]\leq 0.
 \]
 Similarly, $1=\expec[Y_T X_T] \geq \expec[Y_T \tilde{X}_T]$ yields
\[
 \expec[Y_T (\tilde{X}_T - X_T)]\leq 0.
\]
Summing up the previous two inequalities and using the definition $r_T = X^T_T /\tilde{X}^T_T$ gives
\begin{equation}\label{eq: r^T-1 exp}
 0\geq \expec[(\tilde{X}^{-1}_T - Y_T) (X_T - \tilde{X}_T)] = \expec\bra{\pare{1-\frac{X_T Y_T}{r_T}}(r_T-1)}.
\end{equation}
Observe that $(1- X_T Y_T /r_T)(r_T-1)\leq 0$ when $X_T Y_T \leq r_T \leq 1$ or $1\leq r_T \leq X_T Y_T$. Therefore,
\begin{align*}
& \expec\bra{\pare{\pare{1-\frac{X_T Y_T}{r_T}}(r_T-1)}_-}\\
  \leq\ & \expec\bra{\pare{1-\frac{X_T Y_T}{r_T}}(1-r_T) \, \indic_{\{X_T Y_T \leq r_T \leq 1 \, \text{or} \, 1\leq r_T \leq X_T Y_T\}}}.
\end{align*}
As $\expec[((1-X_TY_T/r_T)(r_T-1))_+] \leq \expec[((1-X_TY_T/r_T)(r_T-1))_-]$ by \eqref{eq: r^T-1 exp}, it follows that
\begin{equation}\label{eq: abs XY}
\begin{split}
 &\expec\bra{\left|1-\frac{X_T Y_T}{r_T}\right| |r_T-1|}\\
  \leq\ & 2 \, \expec\bra{\pare{1-\frac{X_T Y_T}{r_T}}(1-r_T) \, \indic_{\{X_T Y_T \leq r_T \leq 1 \, \text{or} \, 1 \leq r_T \leq X_T Y_T\}}}.
  \end{split}
\end{equation}
Now, estimate the expectation on the right-hand side. On the set $\{X_T Y_T\leq r_T \leq 1\}$ we have $(1-X_T Y_T/r_T)(1-r_T) \leq (1- X_T Y_T)^2$, so that
\begin{equation*}
\begin{split}
 &\expec\bra{\pare{1-\frac{X_TY_T}{r_T}}(1-r_T)\, \indic_{\{X_T Y_T \leq r_T \leq 1\}}}\\
 \leq\  &\expec\bra{(1- X_T Y_T)^2 \, \indic_{\{X_T Y_T \leq 1\}}}\\
\leq\ &\expec[(1-X_T Y_T)^2 \, \indic_{\{X_T Y_T\leq 1-\epsilon\}}] + \expec[(1-X_T Y_T)^2 \, \indic_{\{1-\epsilon < X_T Y_T \leq 1\}}]\\
  \leq\  &\prob(X_T Y_T \leq 1-\epsilon) + \epsilon^2, \quad \mbox{for any $\epsilon>0$}.
\end{split}
\end{equation*}
Lemma \ref{lem: lim XY} shows that the first term on the right vanishes as $T\rightarrow \infty$. As $\epsilon$ was chosen arbitrarily, this yields
\begin{equation}\label{eq: est XY<r<1}
 \limT \expec\bra{\pare{1-\frac{X_TY_T}{r_T}}(1-r_T)\, \indic_{\{X_T Y_T \leq r_T \leq 1\}}} =0.
\end{equation}
On $\{1\leq r_T \leq X_T Y_T\}$ we have $X_T Y_T/r_T + r_T \geq 2$ and in turn
$$(1- X_T Y_T/ r_T)(1-r_T)\leq X_T Y_T -1.$$
As a consequence:
\begin{align*}
 &\expec\bra{\pare{1-\frac{X_T Y_T}{r_T}}(1-r_T) \,  \indic_{\{1\leq r_T \leq X_T Y_T\}}}\\
 =\  &\expec\bra{\pare{1-\frac{X_T Y_T}{r_T}}(1-r_T)\, \indic_{\{1\leq r_T \leq X_T Y_T, X_T Y_T \leq 1+\epsilon\}}}\\
 &+  \expec\bra{\pare{1-\frac{X_T Y_T}{r_T}}(1-r_T)\, \indic_{\{1\leq r_T \leq X_T Y_T, X_T Y_T > 1+\epsilon\}}}\\
\leq\  &\expec\bra{(1-X_T Y_T)^2\, \indic_{\{1\leq r_T \leq X_T Y_T, X_T Y_T \leq 1+\epsilon\}}} + \expec\bra{(X_T Y_T -1) \,\indic_{\{X_T Y_T > 1+\epsilon\}}}\\
\leq\  &\epsilon^2 + \expec[|X_T Y_T-1|].
\end{align*}
As $\epsilon$ was chosen arbitrarily and the second term on the right-hand side converges to $0$ by Corollary~\ref{cor: XY conv}, we obtain
\begin{equation}\label{eq: est 1<r<XY}
 \limT \expec\bra{\pare{1-\frac{X_T Y_T}{r_T}}(1-r_T) \, \indic_{\{1\leq r_T \leq X_T Y_T\}}} =0.
\end{equation}
Together, \eqref{eq: abs XY}, \eqref{eq: est XY<r<1} and \eqref{eq: est 1<r<XY}  yield the assertion.
\end{proof}

We are now ready to complete the proof of Proposition \ref{prop: ratio X} iii).

\begin{proof}[Proof of Proposition \ref{prop: ratio X} iii)]
 Lemma \ref{lem: r est} implies
 $$\prob-\limT (1- X_T Y_T/r_T) (1-r_T) =0.$$
Combined with Lemma \ref{lem: lim XY}, this yields the assertion $\prob-\limT r_T=1$.
\end{proof}

\section{Analysis of the Counterexample}\label{sec:analysis}
In this section, we provide a detailed analysis of the counterexample from Section~\ref{sec:example}. Recall that $U(x) = x^p/p$, $p<0$ for sufficiently large $x$ and $U(x) = x^{p^*}/{p^*}$, $p^*<p-1$, for $x\leq 1$. In what follows, we will show that if \eqref{eq: para rest} is satisfied, then
\begin{equation}\label{eq: ratio exp}
 \lim_{T\rightarrow \infty} \frac{\expec[U(\tilde{X}_T)]}{\expec[\tilde{U}(\tilde{X}_T)]} =\infty \quad \text{and} \quad \limsup_{T\rightarrow \infty} \frac{\expec[U(X_T)]}{\expec[\tilde{U}(\tilde{X}_T)]}\leq 1.
\end{equation}
Hence
\begin{equation}\label{eq: ratio conv 0}
 \lim_{T\rightarrow \infty} \frac{\expec[U(\tilde{X}_T)]}{\expec[U(X_T)]} =\infty,
\end{equation}
which confirms \eqref{eq: ce conv 0} as $U$ is regularly varying at infinity (cf.\ Lemma 3.1). Indeed, set $a_T = U^{-1}(\expec[U(\tilde{X}_T)])$ and $b_T = U^{-1}(\expec[U(X_T)])$.
If $\limsup_{T\rightarrow \infty} \expec[U(\tilde{X}_T)]$ is bounded away from zero, then $\limsup_{T\rightarrow \infty}a_T<\infty$. However, $U(\infty) =0$ and $\lim_{T\rightarrow \infty} S^0_T=\infty$ yield $0\geq \expec[U(X_T)] \geq \expec[U(S^0_T)]\rightarrow 0$ as $T\rightarrow \infty$, hence $\lim_{T\rightarrow \infty}b_T =\infty$. Therefore $\lim_{T\rightarrow \infty} a_T/b_T =0$ holds. When $\lim_{T\rightarrow \infty} \expec[U(\tilde{X}_T)] =0$, then $\lim_{T\rightarrow \infty} a_T=\infty$. If there exists $\delta>0$ such that $\liminf_{T\rightarrow \infty}a_T/b_T \geq \delta$, then
\[
 \limsup_{T\rightarrow \infty}\frac{U(a_T)}{U(b_T)} \leq \limsup_{T\rightarrow \infty}\frac{U(\delta b_T )}{U(b_T)}= \delta^p<\infty
\]
because $U \leq 0$, which contradicts \eqref{eq: ratio conv 0}.

To prove the first convergence in \eqref{eq: ratio exp}, it suffices to show
\begin{equation}\label{eq: tXp*/tXp}
  \lim_{T\rightarrow \infty} \frac{\expec[\tilde{X}^{p^*}_T \,\indic_{\{\tilde{X}_T\leq 1\}}]}{\expec[\tilde{X}^p_T]} = \infty.
\end{equation}
 In the Black-Scholes model, the optimal risky weight for power utility $\tilde{U}(x)=x^{p}/p$ is $\tilde{\pi}= \frac{1}{1-p} \frac{\mu}{\sigma^2}$. The associated wealth process starting from unit initial capital is
\begin{equation*}\label{eq: tX}
 \tilde{X}_T = \exp\pare{\pare{r+ \frac{1-2p}{2(1-p)^2} \frac{\mu^2}{\sigma^2}}T + \frac{1}{1-p} \frac{\mu}{\sigma} W_T}.
\end{equation*}
Straightforward calculations show
\begin{equation}\label{eq: txp*}
 \frac{\tilde{X}^{p^*}_T}{\expec[\tilde{X}_T^p]} = \exp\pare{(p^*-p) rT + \pare{\frac{p^*(1-2p+p^*)}{2(1-p)^2} + \frac12 q} \frac{\mu^2}{\sigma^2} T} \mathcal{E}\pare{\frac{p^*}{1-p} \frac{\mu}{\sigma} W_T},
\end{equation}
where $q=p/(p-1)$. Define a new probability measure $\qprob^*$ via
\[
 \frac{d\qprob^*}{d\prob}|_{\F_T} = \mathcal{E}\pare{\frac{p^*}{1-p} \frac{\mu}{\sigma} W_T}.
\]
Then, $W^*_t = W_t - \frac{p^*}{1-p}\frac{\mu}{\sigma} t$ is a $\qprob^*$-Brownian motion on $[0,T]$. As a result,
\begin{align*}
 &\expec^{\prob}\bra{\mathcal{E}\pare{\frac{p^*}{1-p} \frac{\mu}{\sigma} W_T} \,\indic_{\{\tilde{X}_T\leq 1\}}}\\
  =\ & \qprob^*\pare{\frac{1}{1-p}\frac{\mu}{\sigma} W_T \leq - \pare{r+ \frac{1-2p}{2(1-p)^2} \frac{\mu^2}{\sigma^2}} T}\\
 =\ & \qprob^*\pare{\frac{1}{1-p}\frac{\mu}{\sigma} W^*_T \leq -rT - \frac{1-2p+2p^*}{2(1-p)^2} \frac{\mu^2}{\sigma^2}T}\\
 = \ & \qprob^*\pare{\frac{1}{1-p}\frac{\mu}{\sigma} \frac{W^*_T}{\sqrt{T}} \leq -r\sqrt{T} - \frac{1-2p+2p^*}{2(1-p)^2} \frac{\mu^2}{\sigma^2}\sqrt{T}}.
\end{align*}
Observe that
\[
 -r-\frac{1-2p+2p^*}{2(1-p)^2}\frac{\mu^2}{\sigma^2} > -r + \frac{1}{2(1-p)^2} \frac{\mu^2}{\sigma^2} >0,
\]
where the first inequality follows from $p^*<p-1$ and the second inequality holds due to $\mu^2/\sigma^2 > 2(1-p)^2 r$ in \eqref{eq: para rest}. As a consequence:
\begin{equation}\label{eq: Q* lim}
\begin{split}
 &\lim_{T\rightarrow \infty} \expec^{\prob}\bra{\mathcal{E}\pare{\frac{p^*}{1-p} \frac{\mu}{\sigma} W_T} \,\indic_{\tilde{X}_T\leq 1}}\\
  =\ & \lim_{T\rightarrow \infty}  \qprob^*\pare{\frac{1}{1-p}\frac{\mu}{\sigma} \frac{W^*_T}{\sqrt{T}} \leq -r\sqrt{T} - \frac{1-2p+2p^*}{2(1-p)^2} \frac{\mu^2}{\sigma^2}\sqrt{T}}= 1.
\end{split}
\end{equation}
For the exponential factor on the right-hand side of \eqref{eq: txp*}, note that
\begin{align*}
 \frac{p^*(1-2p + p^*)}{2(1-p)^2} + \frac12 q >0
\end{align*}
because $p^*(1-2p+p^*)$ is strictly decreasing in $p^*$ when $p^*<p-1$. If \eqref{eq: para rest} is satisfied, it follows that
\begin{equation}\label{cond2}
 (p^*-p) r + \pare{\frac{p^*(1-2p+p^*)}{2(1-p)^2} + \frac12 q}\frac{\mu^2}{\sigma^2} = (p^*-p) \pare{r+\frac{p^*+1-p}{2(1-p)^2} \frac{\mu^2}{\sigma^2}}>0,
\end{equation}
so that the exponential term on the right-hand side of \eqref{eq: txp*} diverges as $T\rightarrow \infty$. Therefore, \eqref{eq: tXp*/tXp} is obtained after taking into account \eqref{eq: Q* lim}.

Now, consider the second convergence in \eqref{eq: ratio exp}. Let us first prove
 \begin{equation}\label{eq: ratio exp small x}
  \lim_{T\rightarrow \infty} \frac{\expec[U(X_T) \, \indic_{\{X_T\leq M\}}]}{\expec[\tilde{X}^p/p]} =0.
 \end{equation}
 As the market is complete, there exists a common stochastic discount factor $Y$ and $y^T, \tilde{y}^T>0$ such that
 $U'(X_T) = y^T Y_T$ and $\tilde{X}^{p-1}_T = \tilde{y}^T Y_T$.
 Hence, $X_T = I\pare{\frac{y^T}{\tilde{y}^T} \tilde{X}_T^{p-1}}$, where $I= (U')^{-1}$. Define $U^*(T,x) := U\pare{I\pare{\frac{y^T}{\tilde{y}^T} x^{p-1}}}$. Recall that $U(x) = x^{p^*}/p^*$ for small $x$. Therefore,
 \[
  U^*(T, x)= \frac{1}{p^*} \pare{\frac{y^T}{\tilde{y}^T}}^{\frac{p^*}{p^*-1}} x^{\frac{p-1}{p^*-1}p^*}, \quad \mbox{for small } x.
 \]
 On the other hand, there exists $T_0$ such that $1/2\leq y^T/\tilde{y}^T \leq 2$ for any $T\geq T_0$ due to Lemma \ref{lem: ratio Y}.
 As a result, \eqref{ass: U/tU'} is satisfied when $U$ is replaced by $U^*$, i.e.
 \[
  \liminf_{x\downarrow 0}\frac{U^*(T, x)}{x^{p-1}} = \liminf_{x\downarrow 0} \frac{1}{p^*} \pare{\frac{y^T}{\tilde{y}^T}}^{\frac{p^*}{p^*-1}} x^{\frac{p-1}{p^*-1}} = 0, \quad \mbox{for } T>T_0,
 \]
 where the second convergence holds because $p^*<p-1<0$.
 It then follows from Remark \ref{rem: Utx/txp} that
 \begin{align}
   \lim_{T\rightarrow \infty} \frac{\expec[U(X_T) \, \indic_{\{X_T <M\}}]}{\expec[\tilde{X}^p_T/p]} &= \lim_{T\rightarrow \infty} \frac{\expec\bra{U^*(T, \tilde{X}_T)\, \indic_{\{\tilde{X}_T \leq (U'(M) \tilde{y}^T/y^T)^{1/(p-1)}\}}}}{\expec[\tilde{X}_T^p/p]}\notag\\
   &=0, \label{eq: Utx conv}
 \end{align}
for any $M>0$. On the other hand, fix $a>0$. For any $\epsilon>0$ there is $M_{a,\epsilon}$ such that $1-\epsilon \leq U'(x)/(a+x)^{p-1} \leq 1+\epsilon$ for $x\geq M_{a,\epsilon}$. Then \eqref{eq: U bdd p<0 a} follows, and the second inequality therein yields
 \begin{equation}\label{eq: Ux/tx}
  \frac{\expec[U(X_T)]}{\expec[\tilde{X}^p_T /p]} \leq \frac{\expec[U(X_T) \, \indic_{\{X_T <M_{a,\epsilon}\}}]}{\expec[\tilde{X}^p_T/p]} + (1+\epsilon) \frac{\expec\bra{(a+X_T)^p\,\indic_{\{X_T \geq M_{a,\epsilon}\}}}}{\expec[\tilde{X}^p_T]},
 \end{equation}
 where the first term on the right-hand side vanishes as $T\rightarrow \infty$ due to \eqref{eq: Utx conv}. For the second term, use \eqref{eq: exp X/y} with $\lim_{T\rightarrow \infty} y^T/\tilde{y}^T=1$ and $\tilde{y}^T = \expec[\tilde{X}_T^p]$, to obtain
 \[
  \limsup_{T\rightarrow \infty} \frac{\expec\bra{(a+X_T)^p\,\indic_{\{X_T \geq M_\epsilon\}}}}{\expec[\tilde{X}^p_T]} \leq \frac{1}{1-\epsilon}.
 \]
In summary, the estimates for the two terms on the right side of \eqref{eq: Ux/tx} yield
 \[
  \limsup_{T\rightarrow \infty} \frac{\expec[U(X_T)]}{\expec[\tilde{X}^p_T/p]} \leq \frac{1+\epsilon}{1-\epsilon}.
 \]
 which confirms the second convergence in \eqref{eq: ratio exp} as $\epsilon$ was chosen arbitrarily.

\bibliographystyle{abbrvnat}
\bibliography{biblio}

\end{document}